\documentclass[12pt,a4paper]{article}

\usepackage[centertags]{amsmath}
\usepackage{amsfonts,amsthm,amssymb}
\usepackage{amssymb}
\usepackage{amsmath}
\usepackage{url}
\usepackage{soul}
\usepackage{graphicx}
\usepackage{accents}
\usepackage{booktabs}
\usepackage{appendix}
\usepackage[hmargin=3.175cm,vmargin=3.175cm]{geometry}
\usepackage{etoolbox}
\usepackage{xparse}
\usepackage{enumitem}
\usepackage{authblk}
\usepackage{color}
\usepackage{soul}
\usepackage{chngcntr}
\usepackage{authblk}

\newcommand{\ubar}[1]{\underaccent{\bar}{#1}}
\DeclareDocumentCommand{\publicBelief}{O{\mu}}{#1}
\DeclareDocumentCommand{\privateBelief}{O{p}}{#1}
\DeclareDocumentCommand{\signalupdate}{O{q}}{#1}
\DeclareDocumentCommand{\signal}{O{s}}{#1}
\DeclareDocumentCommand{\signalsSet}{O{S}}{#1}
\DeclareDocumentCommand{\state}{O{\omega}}{#1}
\DeclareDocumentCommand{\statesSet}{O{\Omega}}{#1}
\DeclareDocumentCommand{\price}{O{\tau}}{#1}
\DeclareDocumentCommand{\action}{O{a}}{#1}
\DeclareDocumentCommand{\limitParam}{O{\alpha}}{#1}
\DeclareDocumentCommand{\deterrencePrice}{O{\price} O{d}}{#1^#2}
\DeclareDocumentCommand{\LLR}{O{x}}{\log(\frac{#1}{1-#1})}
\DeclareDocumentCommand{\llr}{O{x}}{l\left(#1\right)}
\DeclareDocumentCommand{\lBound}{O{\limitParam} O{\publicBelief}}{\ubar{#1}_{#2}}
\DeclareDocumentCommand{\uBound}{O{\limitParam} O{\publicBelief}}{\bar{#1}_{#2}}
\DeclareDocumentCommand{\eqPrice}{O{\price}}{#1^{*}}
\newcommand{\posterior}{\privateBelief_{\mu}(s)}

\newcommand{\Diff}[2]{\frac{\partial #1}{\partial #2}}

\newcommand{\gor}{\rightarrow}

\newtheorem{lemma}{Lemma}
\newtheorem{proposition}{Proposition}
\newtheorem{observation}{Observation}
\newtheorem{theorem}{Theorem}

\newtheorem{corollary}{Corollary}

\newtheorem{assumption}{Assumption}
\theoremstyle{definition}
\newtheorem{definition}{Definition}

\newtheorem{remark}{Remark}[section]

\newtheorem*{cor1}{Corollary \ref{cor:asymptotic_learning}}
\newtheorem*{cor2}{Corollary \ref{lem:AL_consumer_buys_superior}}
\newlist{secenum}{enumerate}{10}
\setlist[secenum]{label=\thesection.\arabic*,leftmargin=*}

\newtheorem{innercustomgeneric}{\customgenericname}
\providecommand{\customgenericname}{}
\newcommand{\newcustomtheorem}[2]{%
	\newenvironment{#1}[1]
	{%
		\renewcommand\customgenericname{#2}%
		\renewcommand\theinnercustomgeneric{##1}%
		\innercustomgeneric
	}
	{\endinnercustomgeneric}
}

\newcustomtheorem{customthm}{Theorem}
\newcustomtheorem{customlemma}{Lemma}
\newcommand{\blocktheorem}[1]{%
	\csletcs{old#1}{#1}
	\csletcs{endold#1}{end#1}
	\RenewDocumentEnvironment{#1}{o}
	{\par\addvspace{1.5ex}
		\noindent\begin{minipage}{\textwidth}
			\IfNoValueTF{##1}
			{\csuse{old#1}}
			{\csuse{old#1}[##1]}}
		{\csuse{endold#1}
		\end{minipage}
		\par\addvspace{1.5ex}}
}

\blocktheorem{customlemma}

	\definecolor{ao}{rgb}{0.0, 0.5, 0.0}

\usepackage{hyperref}
\usepackage{subcaption}

\begin{document}
\title{The Implications of Pricing on Social Learning.}
\author{Itai Arieli }
\author{Moran Koren}
\author{Rann Smorodinsky}

\affil{\small Faculty of Industrial Engineering, Technion\textemdash Israel Institute of Technology. \thanks{Research supported by GIF research grant no. I-1419-118.4/2017, ISF grant 2018889, Technion VPR grants, the joint Microsoft-Technion e-Commerce Lab, the Bernard M. Gordon Center for Systems Engineering at the Technion, and the TASP Center at the Technion.}}

\renewcommand\Authands{, and }
\renewcommand\footnotemark{}
\maketitle
\abstract{\noindent

We study the implications of endogenous pricing for learning and welfare in the classic herding model (\cite{Banerjee1992,Bikhchandani1992}). When prices are determined exogenously, it is known that learning occurs if and only if signals are unbounded (see Smith and S\o rensen \cite{Smith2012}). By contrast, we show that learning can occur when signals are bounded as long as non-conformism among consumers is scarce. More formally, learning happens if and only if signals exhibit the  \textit{vanishing likelihood} property introduced bellow.
We discuss the implications of our results for potential market failure in the context of Schumpeterian growth with uncertainty over the value of innovations.


\medskip


\medskip
\section{Introduction}
In many markets of substitute products,
the value of the various alternatives may depend on some unknown variable. This may take the form of some future change in regulation, a technological shock, environmental developments, or prices in related upstream markets. Although this information is unknown, individual consumers may receive some private information about these fundamentals. We ask whether, in such an environment, markets aggregate information correctly and the ex-post superior product  eventually dominates the market.



In this work we focus on the role of  social learning in such environments. We study whether the learning process guarantees an efficient outcome.
 We isolate the role of learning by introducing a simple duopoly model of common value in which consumers, with a unit demand,  choose between two substitute products, each with zero marginal cost of production.
The timing of the interaction is as follows. Nature randomly chooses one of two states of nature, and so determines the identity of the firm with the superior product. At each stage the two firms observe the entire history of the market -- past prices and consumption decisions -- and simultaneously set prices.
Thereafter, a single consumer  arrives and receives a private signal regarding the state of nature. The consumer makes his consumption decision based on his signal, the pair of prices for each product, and the entire history of prices and consumption decisions. The consumer can also choose to opt out of buying any product.
Our main goal is to identify conditions under which \emph{asymptotic learning} holds; that is, information is fully aggregated in the market asymptotically.
%

When prices are set exogenously and fixed throughout, the above model is exactly the standard  herding model \cite{Banerjee1992,Bikhchandani1992}. In that model, as shown by Smith and S\o rensen \cite{Smith2012}, the characterization of asymptotic learning crucially depends upon the quality of agents' private signals. In particular, one must distinguish between two families of signals:
bounded versus unbounded. In the unbounded case the private beliefs of the agents are, with positive probability, arbitrarily close to zero and one. Therefore, no matter how many people herd on one of the alternatives, there is always a positive probability that the next agent will receive a  signal that will make him break away from the herd toward the other alternative. This property, as shown by  Smith and S\o rensen \cite{Smith2012}, entails asymptotic learning.
The same logic applies in our model as well: when signals are unbounded, even if the prior is extremely in favor of one product, with positive probability there will be a consumer who gets a sufficiently strong  signal,that tilts the consumption decision toward the a priori inferior product; thus, under strategic pricing and unbounded signals, asymptotic learning holds.

The learning results in our model depart from those of the canonical model when signals are bounded. In the herding model there is always a positive probability that the suboptimal alternative will eventually be chosen by all agents. However, intuition suggests that when prices are endogenized they serve to prevent such a herding phenomenon. Hypothetically, once society stops learning and a herd develops on the product of one firm, the other firm will lower its product price to attract new consumers and learning will not cease. It turns out that this intuition, although not entirely correct, does have some merit. In order for the intuitive argument to hold, signals must exhibit a property referred to here  as \textit{vanishing likelihood}.

When signals are bounded the posterior belief of any agent, given his signal, is bounded away from zero and one for any interior prior. The proportion of agents whose posterior lies within $\varepsilon$ of the posterior distribution's boundaries obviously shrinks to zero as $\varepsilon$ goes to zero. We say that signals exhibit  \textit{vanishing likelihood} if
the \emph{density} of consumers at the posterior belief distribution's boundaries goes to zero.

Consumers who receive signals that induce such extreme posterior beliefs are those consumers who are likely to go against a herd and purchase the less popular product. We refer to such consumers as {\it nonconformists}. With this interpretation in mind the property of ``vanishing likelihood'' serves as a measure of the prevalence of nonconformism. More particularly, we associate vanishing likelihood with a negligible level of nonconformism while signals that do not exhibit vanishing likelihood are associated with significant nonconformism.

When society herds, each agent follows in the footsteps of  his predecessors and thus, intuitively, one expects nonconformism (when signals do not exhibit vanishing likelihood) to induce learning.
Our main result shows that the opposite occurs:  in the presence of strategic pricing asymptotic learning holds if and only if signals have the vanishing likelihood property.



\subsection{Schumpeterian Growth}\label{sec:schumpeter}
It is widely agreed that innovation and the evolution of technology constantly propel the economy forward. New technologies replace older ones and may improve product quality, reduce production costs, and often completely disrupt an industry.

However, not every innovation entails improvement. Arguably, innovations that do not entail improvement  will naturally be driven out of the market and only those that do will prevail. This argument forms the basis of the  evolutionary economics literature that dates back to Marx, Veblen, and Schumpeter.

In his seminal work, Schumpeter \cite{Schumpeter1942} described the process of economic growth, which he refers to as ``Capitalism,"  as an evolutionary process that is shaped by ``gales of creative destruction." Some  of his contemporaries had argued that large and profitable firms are the source of innovation and so regulation protecting them was essential to R\&D investments. By contrast, Schumpeter argued that incumbent firms, anticipating innovation by potential entrants, invest in R\&D to stay ahead of the game. Therefore such regulation is unwarranted, and may even be detrimental. However, such profitable incumbents may also use their power to drive innovation away by lowering prices. This is true in particular when it is hard to identify which innovation constitutes an improvement and which does not.

Does the evolutionary process guarantee that the economy will successfully separate the wheat from the chaff? This question becomes more acute with the accelerated pace of innovation witnessed in the past two decades \cite{OECD2015}.

Our theoretical results shed light on this issue and relate the outcome of the evolutionary process to the market structure. Our model shows that whenever the proportion of nonconformist consumers (often referred to as ``early adopters'' in the context of technological revolutions) is insignificant (a phenomenon captured by the technical notion of vanishing likelihood),  the evolutionary process successfully sieves the better technologies. However, whenever this proportion is significant then the evolutionary process may fail and policies to support entry may be warranted in order to sustain Schumpeterian growth.

In Section \ref{sec:examples} we discuss two case studies from the late $90$ related to technological innovations. In both cases the innovative technology turned out to be the better than the incumbent one but only in one of these cases was it adopted by society. Our model provides an explanation as it distinguishes between the two cases, based on whether the underlying information structure complied with the vanishing likelihood condition.

\subsection{Related Literature}
Our work primarily contributes to the work of Bikhchandani, Hirshleifer, and Welch \cite{Bikhchandani1992} and Banerjee \cite{Banerjee1992}, who introduced
models of social learning with agents who act sequentially. Their primary contribution was to point out the possibility of information cascades and market failure when signals are bounded. Smith and S\o rensen \cite{Smith2012} characterize the information structure that induces such a potential market failure.
In these and  many of the follow-up papers, prices are assumed exogenous and fixed throughout. The primary departure of our model from this line of research is that our model incorporates endogenous pricing. 
We associate a favorite firm (product) with each state of nature and allow for the firms to set prices dynamically, based on the information available in the market.

 Avery and Zemsky \cite{Avery1998} incorporate dynamic pricing into herding models. They consider a single firm whose product value is associated with an (unknown) state of nature. Instead of offering the product at a fixed price, as in the earlier papers, they assume that the price is set dynamically to be the expected value of the product conditional on all the publicly available information. Since their primary interest is to study financial markets, they assume that there is a market maker that  uses all the publicly available data to set prices. By contrast, we assume that the firms themselves set the prices.

A model that is reminiscent of our model is that of  Bose et al. \cite{Bose2006, Bose2008} who study a monopoly, with a good of uncertain quality, that dynamically adjusts prices to compete against an outside option. Consumers arrive sequentially and make a consumption decision based on their predecessors' decisions, past prices, and an additional private signal. In \cite{Bose2006} the authors restrict attention to signal structures with finitely many signals  whereas in \cite{Bose2008} they further restrict attention to symmetric binary signal structures.
In both models it is shown that herding is inevitable and if the public belief is sufficiently in favor of the monopoly, then the monopoly will price low enough to attract all consumers, regardless of their realized signal. As vanishing likelihood never holds when the signal space is finite (see the discussion in Section \ref{sec:gen_sig}), the results of Bose et al.,  albeit in a monopoly framework, hold in our duopolistic model  (see the additional discussion in  Section 5.3). Restricting  the analysis to finite signal spaces would not allow the unraveling of the vanishing likelihood property that  we introduce and that characterizes learning when prices are endogenous.


Moscarini and Ottaviani \cite{Moscarini1997} study the duopoly case and their paper focuses on a single-stage interaction with two firms and a single knowledgeable consumer. In fact, it is exactly the model of the stage game ($\Gamma(\mu)$) we study in Section \ref{section_Gamma}, except that their model is restricted to a binary and symmetric signal space. Unsurprisingly, whenever the prior belief is above (or below) some threshold, all equilibria in their model form a \emph{deterrence equilibrium} (see definition \ref{def:DE}), where one firm prices out the other firm. Clearly, the emergence of a deterrence equilibrium implies that learning stops in the repeated model. The authors go on to provide comparative statics over the threshold public belief for which learning stops as a function of the informativeness of the signal (here is where they leverage the restricted signal space). As signals become more informative the thresholds move to the extremes.

Our main result on the one-shot game, Theorem \ref{thm:SSG_results}, argues that learning stops whenever the \emph{vanishing likelihood} condition does not hold. As this condition can never hold for a finite signal space the result in \cite{Moscarini1997} follows as a corollary. The main take-home message from comparing our work with \cite{Moscarini1997} is that  learning is  determined not by the level of informativeness of the signals but rather by the vanishing likelihood condition.  In particular, signal distributions that satisfy vanishing likelihood need not be highly informative. The restriction to a binary model, in this case, is misleading. A similar distinction is valid in the monopoly setting of Bose et al. \cite{Bose2006, Bose2008}.

Mueller-Frank \cite{Mueller-frank2012,Mueller-frank2016} introduces a pair of models with dynamic pricing of a monopoly \cite{Mueller-frank2016} and a duopoly \cite{Mueller-frank2012}. The model is very similar to ours with the distinction that the firms have the informational advantage and know the true state of the world.%
\footnote{When firms have an informational advantage the equilibrium analysis, as Mueller-Frank points out,  crucially hinges on consumers' off-equilibrium beliefs.  This is not the case in our model, which consequently allows for robust observations.}
Mueller-Frank does not characterize the informational conditions that entail learning, as we do. Rather, he studies the connection between welfare and learning and shows that learning is not sufficient for welfare maximization (see Corollary \ref{cor:asymptotic_learning}). It is worth noting that in our model, by contrast, learning is necessary and sufficient for welfare maximization.

The paper is organized as follows. Section \ref{section:main results} presents the model and the
main theorem for the case where firms are myopic. Section \ref{section:proof} provides the proof
of the main result. Section \ref{sec:non-myopic} is an extension of our model to the case where
firms are farsighted. Section \ref{section:extensions} informally discusses some extensions and Section \ref{section:discussion} concludes.

\section{Social Learning and Myopic Pricing}\label{section:main results}

Our model comprises a countably infinite number of consumers, indexed by $t\in\mathbb{N}$, and two firms: Firm $0$ and Firm $1$.
There are two states of nature $\Omega = \{0,1\}$. In state $\omega$, firm $\omega\in\{0,1\}$ produces the superior product. We normalize the value of the superior product to $1$ and the value  of the inferior product to $0$. In every time period $t$ the two firms first set (non-negative) prices
$(\tau_0^t,\tau_1^t)\in[0,1]^2$ for their products.  Then consumer $t$ receives a private signal and must decide whether to buy product $0$, product $1$, or neither product. Formally, the action set of every consumer is $A=\{0,1,e\}$, where the action $a$ corresponds to the decision to buy from firm $i$ and the action $a=e$ corresponds to the decision to exit the market and not to buy either product.  The payoff of every consumer $t$, given the price vector $(\tau_0,\tau_1)$ as a function of the realized state $\omega$, is
\begin{equation}\label{eq:util_def}
u(\action,\tau_0,\tau_1,\omega)=\begin{cases}
0\mbox{ if }&\action=e \\
1-\tau_a\mbox{ if }&\action=\state\\
-\tau_a&\mbox{ otherwise.}
\end{cases}
\end{equation}

For simplicity we assume that both firms have no marginal cost of production.
Hence, firm  $i$'s stage payoff, given a price vector $(\tau_0,\tau_1)$, can be described as a function of the consumer's decision as follows:
\begin{equation}\label{eq:profit_def}
\pi_i(\action,\tau_0,\tau_1,\omega)=\begin{cases}
\tau_i  \mbox{ if }&\action = i \\
0 \mbox{ if }&\mbox{ otherwise.}
\end{cases}
\end{equation}

We assume that the state $\omega$ is drawn at stage $t=0$ according to a commonly known prior distribution,
such that $P(\omega = 0)=\mu_0 =1- P(\omega = 1)$.
The state $\omega$ is unknown to both the firms and the consumers. Each consumer $t\in\mathbb{N}$ forms a belief on the state using two sources of information: the history of prices and actions, $h_t\in H_{t} = ([0,1]^{2}\times \{0,1,e\})^{t-1}$, and a private signal
$s_t \in S$ (where $S$  is some abstract measurable signal space).\footnote{An alternative model is to assume that consumers do not observe prices. In this alternative formulation our results go through when restricting attention to pure equilibria. This follows from the observation that the consumption history determines the corresponding equilibrium prices at each stage.} The firms observe only the realized history $h_t\in H_t$ at every time $t$ and receive no private information.  Conditional on the state  $\omega$, signals are independently
drawn according to a probability measure $F_\omega$. We refer to  the tuple $(F_0,F_1,S)$ as a \textit{signal structure}. 
 We assume throughout that $F_{0}$ and $F_{1}$ are mutually absolutely continuous with respect to each other.\footnote{$F_0$ and $F_1$ are mutually absolutely continuous whenever $F_0(\hat S)>0 \iff F_1(\hat S)>0$ for any measurable set $\hat S \subset S$. Note that with this assumption the probability of a fully revealing signal, for which the posterior probability is either $0$ or $1$, is zero.}
The prior $\mu_0$ and the functions $F_{0}$ and $F_{1}$ are common knowledge among consumers and firms.

%

We let $H=\cup_{t\geq 1} H_t$ be the set of all finite histories and let $H_\infty=([0,1]^{2}\times \{0,1,e\})^\infty$ be the set of all infinite histories.
We let $\mathcal{A}\subset \Delta(\{0,1,e\})^{[0,1]^2\times S}$ be the set of decision rules for the consumer; i.e., $\mathcal{A}$ is the set of all measurable functions that map pairs consisting of a price vector and a signal to a consumption decision.
A \textit{(pure) strategy for consumer $t$}  is a measurable function $\sigma_t:H_{t}\rightarrow \mathcal{A}$ that maps   every history $h_t\in H_t$ and signal $s_t\in S$ to a decision rule.
We denote by $\bar{\sigma}=(\sigma_t)_{t\geq 1}$ a pure strategy profile for the consumers. We can view $\bar{\sigma}$ as a function $\bar{\sigma}:H\rightarrow \mathcal{A}.$ 
 A (behavioral) strategy for firm $i$ is a (measurable) mapping $\bar{\phi}_{i}:H\rightarrow \Delta([0,1])$. We note that the strategy profile $(\bar{\phi_0},\bar{\phi_1},\bar{\sigma})$ together with the prior $\mu_0$ and the signal structure $(F_0,F_1,S)$ induce a probability distribution $\mathbf{P}_{(\bar{\phi_0},\bar{\phi_1},\bar{\sigma})}$ over $\Omega\times H_\infty \times S^\infty$.

Let $\mu_t=\mathbf{P}_{(\bar{\phi}_0,\bar{\phi}_1,\bar{\sigma})}(\omega=0|h_t)$ be the probability that the state is $0$ conditional on the realized history $h_t\in H$. We call $\mu_t$ \emph{the public belief at time} $t$.
The following observation regarding the sequence of public beliefs, $\{\mu_t\}_{t=1}^{\infty},$ is straightforward.
\begin{observation}\label{obs:martingale}
$\{\mu_t\}_{t=1}^{\infty}$ is a martingale. Thus, by the martingale convergence theorem, it must converge almost surely to a limit random variable $\mu_\infty\in[0,1]$.
\end{observation}

A strategy profile $(\bar{\phi_0},\bar{\phi_1},\bar{\sigma})$ and a history $h_t$ 
 induce both an expected payoff $\Pi^t_i(\tau_0,\tau_1,\bar\sigma|h_t)$ for every firm $i$ and an expected consumer utility  $u_t(\tau_0,\tau_1,\bar\sigma|h_t).$ 
  We can now define the notion of a Bayesian Nash equilibrium for myopic firms.
\begin{definition}\label{def:equilibrium0}
A strategy profile $(\bar{\phi}_0,\bar{\phi}_1,\bar{\sigma})$ constitutes a \emph{myopic Bayesian Nash equilibrium} if for every time $t$ the following conditions hold for almost every history $h_t\in H_t$ that is realized in accordance with $\mathbf{P}_{(\bar{\phi}_0,\bar{\phi}_1,\bar{\sigma})}$:
\begin{itemize}
\item For every $\tau\in[0,1]$,
$$\Pi^t_i(\bar{\phi}_0,\bar{\phi}_1,\bar{\sigma}|h_t)\geq \Pi^t_i(\tau,\bar{\phi}_{-i},\bar{\sigma}_t|h_t).$$
 \item For every price vector  $(\tau_0,\tau_1)\in[0,1]^2$, and every decision rule $\sigma\in\mathcal{A},$
 $$U_t(\tau_0,\tau_1,\bar{\sigma}(h_t)|h_t)\geq U_t(\tau_0,\tau_1,\sigma|h_t).$$
\end{itemize}
\end{definition}

In words, a strategy profile $(\bar{\phi}_0,\bar{\phi}_1,\bar{\sigma})$ constitutes a myopic Bayesian Nash equilibrium if for every time $t$ and $\mathbf{P}_{(\bar{\phi}_0,\bar{\phi}_1,\bar{\sigma})}$ almost every history $h_t\in H_t$, $\bar{\phi}_i(h_t)$ maximizes the conditional expected stage payoff to every firm $i$ and $\bar{\sigma}(h_t)$ maximizes the conditional expected payoff to consumer $t$ with
respect to \emph{every} price vector $(\tau_0,\tau_1)$.

Note that our notion of equilibrium is weaker than the notion of a subgame perfect equilibrium; however, it still eliminates  equilibria with non-credible threats by consumers. One such equilibrium with non-credible threats is the following equilibrium: both firms
ask for price $0$ in every time period. Every consumer $t$ never buys a product (i.e., plays $e$) unless both firms ask for a price of $0,$ in which case he buys product $0$ whenever $\mu_t\geq \frac{1}{2}$ and product $1$ if $\mu_t<\frac{1}{2}$. Note that this equilibrium is sustained by non-credible threats made by the consumer. Such threats are eliminated by the second condition, which requires that conditional on the realized history $h_t$ the decision rule $\bar{\sigma}(h_t)$ be optimal with respect to \emph{every} price vector $(\tau_0,\tau_1)$, and not just with respect to $(\tau_0^t(h_t),\tau_1^t(h_t))$. We note that since our results in the sequel hold for all myopic equilibria they are in particular valid for subgame perfect equilibria.%
\footnote{A natural question is whether restricting attention to the stronger solution concept yields a weaker condition for learning. We conjecture that this is not the case, particularly when firms are myopic.}


As is common in the literature, we define \emph{asymptotic learning} as follows.
\begin{definition}
Fix a signal structure $(F_0,F_1,S).$ Let $\mu_0\in(0,1)$ be the prior and let $(\bar{\phi}_0,\bar{\phi}_1,\bar{\sigma})$ be a strategy profile of the corresponding game. We say that \textit{learning holds for $\mu_0$ and $(\bar{\phi}_0,\bar{\phi}_1,\bar{\sigma})$} if the belief martingale sequence converges almost surely to a point belief assigning probability $1$ to the realized state.  \textit{Asymptotic learning holds for the signal structure $(F_0,F_1,S)$} if learning holds for \emph{every} prior $\mu_0$ and every corresponding myopic Bayesian Nash equilibrium $(\bar{\phi}_0,\bar{\phi}_1,\bar{\sigma})$. By contrast, \textit{asymptotic learning never holds} if for \emph{every} prior $\mu_0\in(0,1)$ and a corresponding myopic Bayesian Nash equilibrium $(\bar{\phi}_0,\bar{\phi}_1,\bar{\sigma}),$
 the limit public belief $\mu_\infty$ lies in $(0,1)$ with positive probability.
\end{definition}
Thus, when asymptotic learning holds, it must be the case that consumers and firms eventually learn the superior product. In our case, even for a very strong public belief in favor of one firm, it is not a priori clear that the strong firm will dominate the market as the weak firm can always lower its price.  However, we show in Lemma \ref{lem:AL_consumer_buys_superior} that the probability of buying from the superior firm converges to one when asymptotic learning occurs. 
Whenever asymptotic learning doesn't hold, only one firm prevails (from some time on, all consumers buy from one firm but are not 100\% certain that it is the superior one). As a result, there is positive probability that the prevailing firm is the inferior one.

\begin{definition}\label{def:G_omega}
	
	Let $f_{\omega}$ denote the Radon--Nikodym derivative of $F_{\omega}$ with respect to the probability measure $\frac{F_0+F_1}{2}$. We consider the random variable $\privateBelief(\signal)\equiv\frac{f_{0}(s)}{f_{0}(s)+f_{1}(s)}$, which is the  posterior probability that $\omega =0$, conditional on the signal $s$, when the prior over $\Omega$ is $(0.5,0.5)$.
Let ${G}_{\state}(x)=F_{\state}(\{s\in\signalsSet|\privateBelief(s)<x\}), \ \omega=0,1$, be the two cumulative distribution correspondences of the random variable $\privateBelief(s)$ induced by the two probability distributions, $F_{\omega}, \ \omega=0,1$, over $S$. Define the bounds $\bar{\limitParam}, \ubar{\limitParam}$ of the signal distribution as follows:
	$\ubar{\limitParam}=\inf_{x \in [0,1]}{ G}_{\omega}(x)>0$ and
	$\bar{\limitParam}=\sup_{x \in [0,1]}{ G}_{\omega}(x)<1$.%
\footnote{Since $F_0$ and $F_1$ are mutually absolutely continuous it follows that $G_1$ and $G_0$ have the same support.}\end{definition}

The main goal of our paper is to provide a characterization of asymptotic learning under strategic pricing in terms of the signal structure $(F_0,F_1,S).$ Such a characterization is provided by Smith and S\o rensen \cite{Smith2012} for the standard herding model where prices are set exogenously.
We start by presenting the formal definition of  \textit{bounded and unbounded signals} due to Smith and S\o rensen \cite{Smith2012}. 
\begin{definition}\label{def:bounded}
The signal structure $(F_0,F_1,S)$ is called \textit{unbounded} if  $\ubar{\limitParam}=0$ and $\bar{\limitParam}=1$. The signal structure $(F_0,F_1,S)$ is  \textit{bounded} if $\ubar{\limitParam}>0$ and $\bar{\limitParam}<1$.
\end{definition}
In words, a signal structure is unbounded if for every $\beta\in(0,1)$ the two sets $\{s:p(s)>\beta\}$ and $\{s:p(s)<\beta\}$ have positive probability under $(F_\omega)_{\omega=0,1}$.  Smith and S\o rensen's characterization shows that in the standard herding model asymptotic learning holds under an unbounded signal structure and fails under a bounded signal structure.

\subsection{Characterization of Asymptotic Learning}
For ease of exposition we make the following assumption on $({ G}_{\state}(x))_{\state=0,1}$. We refer the reader to Section \ref{section:extensions} for the general case.
\begin{assumption}\label{asumption}
We assume that the functions $\{{ G}_{\state}(x)\}_{\state=0,1}$ are differentiable on $(\ubar{\limitParam},\bar{\limitParam})$ with continuous derivatives $({ g}_{\state}(x))_{\state=0,1}:[\ubar{\limitParam},\bar{\limitParam}]\rightarrow\mathbb{R}_+$.
\end{assumption}
\begin{definition}\label{def:VL}
A signal structure $(F_0,F_1,S)$ exhibits \emph{vanishing likelihood} if  ${g}_1(\ubar{\alpha})={g}_0(\bar{\alpha})=0.$
\end{definition}
We  next show how information aggregation depends on the vanishing likelihood property.
The following theorem provides a full characterization of asymptotic learning in our model.
\begin{theorem}\label{thm:social_learning}
If signals are unbounded or if signals are bounded and exhibit vanishing likelihood then asymptotic learning holds. If signals are bounded and do not exhibit vanishing likelihood then asymptotic learning never holds.
\end{theorem}
Although Theorem \ref{thm:social_learning} formally hinges on Assumption 1,  an analog of this theorem carries over completely to a model without such an assumption. This requires an alternative formulation of the vanishing likelihood condition. We elaborate on this in Section \ref{section:extensions}.
\begin{remark}
In fact, the proof of Theorem \ref{thm:social_learning} shows that whenever ${g}_1(\ubar{\alpha})>0$ the limit public belief $\mu_\infty$, conditional on state $\omega=0$, lies in $(0,1)$ \emph{with probability $1$}. Analogously, whenever ${g}_0(\bar{\alpha})>0$ the limit public belief $\mu_\infty$, conditional on the state $\omega=1$, lies in $(0,1)$ \emph{with probability $1$}.
\end{remark}
\section{Proof of the Main Result}\label{section:proof}
In the proof of Theorem \ref{thm:social_learning} we rely on the analysis of the following
three-player stage game $\Gamma(\mu)$. The game comprises two firms and a single consumer and is derived from our sequential game by restricting the game to a single period. That is, in $\Gamma(\mu)$ the state is realized according to the prior $\mu$ (state $0$ is realized with probability $\mu$ and state $1$ with probability $1-\mu$). The two firms post a price simultaneously (possibly at random) and a single consumer receives a private signal in accordance with $(F_0,F_1,S)$ and based on his private signal and the realized vector of prices takes an action $a\in\{0,1,e\}$.
The utility for the consumer is determined by equation \eqref{eq:util_def} and the utility for the firms is determined by equation \eqref{eq:profit_def}.

In the following observation, which is a direct implication of Definition \ref{def:equilibrium0}, we connect the stage game with the sequential model.
\begin{observation}\label{obs:equ}
A strategy profile $(\bar{\phi}_0,\bar{\phi}_1,\bar{\sigma})$ constitutes a myopic Bayesian Nash equilibrium if and only if
for every time $t,$ for $\mathbf{P}_{(\bar{\phi}_0,\bar{\phi}_1,\bar{\sigma})},$ and for almost every history $h_t\in H_t,$ the tuple $(\bar{\phi}_0(h_t),\bar{\phi}_1(h_t),\bar{\sigma}(h_t))$
is a subgame perfect equilibrium (SPE) of $\Gamma(\mu_t)$.
\end{observation}
The strong connection of $\Gamma(\mu)$ to our sequential game allows us to derive some insight into information aggregation from the subgame  perfect equilibrium properties of $\Gamma(\mu)$, which we analyze next.

\subsection{Analysis of $\Gamma(\mu)$}\label{section_Gamma}
We begin by studying the consumer's best-reply strategy in $\Gamma(\mu)$.
We denote the consumer's posterior belief after the consumer observes the signal $s_t=s$  by $p_{\publicBelief}(s)$. It follows readily from Bayes' rule that
\begin{equation}\label{eq:bayesian update}
 \privateBelief_{\publicBelief}(\signal)=\frac{\publicBelief \privateBelief(s)}{\publicBelief \privateBelief(s)+(1-\publicBelief)(1- \privateBelief(s))}.
 \end{equation}
The bounds $\ubar{\limitParam}$ and $\bar{\limitParam}$, together with equation \eqref{eq:bayesian update}, imply that  $p_{\mu}(s) \in [\lBound,\uBound]$ with probability one, where:
\begin{equation}\label{eq:def_bounds}
\lBound=\frac{\publicBelief \ubar{\limitParam}}{\publicBelief \ubar{\limitParam}+(1-\publicBelief)(1- \ubar{\limitParam})}\mbox{ and }
\uBound=\frac{\publicBelief \bar{\limitParam}}{\publicBelief \bar{\limitParam}+(1-\publicBelief)(1- \bar{\limitParam})}.
\end{equation}

Fix a price vector $\tau=(\tau_{0},\tau_{1})$ and note that the consumer optimizes his expected utility against $\tau$ if and only if he follows the following strategy:
\begin{equation}\label{eq:consumer_condition}
\sigma(\mu,s, \price)=\begin{cases}
a=0&\mbox{ if }\privateBelief_{\publicBelief}(\signal)-\price_{0}\geq\max\{(1-\privateBelief_{\publicBelief}(\signal))-\price_{1},0\}\\
a=1&\mbox{ if }(1-\privateBelief_{\publicBelief}(\signal))-\price_{1}\geq\max\{\privateBelief_{\publicBelief}(\signal)-\price_{0},0\}\\
a=e&\mbox{ otherwise}.
\end{cases}
\end{equation}
 We note that, under Assumption \ref{asumption}, in every perfect Bayesian Nash  equilibrium of the game $\Gamma(\mu),$ the strategy $\sigma$ constitutes a unique action for the consumer almost surely.

Note further that every mixed strategy $(\phi_{0},\phi_{1})$ induces two possible market scenarios: \textit{a full market scenario}, where,
under $\sigma,$ the consumer always buys from one of the firms for almost all signal realizations and almost every realized price vector $\tau=(\tau_0,\tau_1)$, and \textit{a non-full market scenario}, where $\sigma(\mu,s, \price)=e$ holds with positive probability.

We can infer from \eqref{eq:consumer_condition} that the consumer buys from Firm $0$ whenever
$\privateBelief_{\publicBelief}(\signal)-\price_{0}\geq (1-\privateBelief_{\publicBelief}(\signal))-\price_{1}$
and the market is full or whenever
$\privateBelief_{\publicBelief}(\signal)-\price_{0}\geq 0$
and the market is not full.


Given a prior $\mu$ and a pair of prices $(\tau_0,\tau_1)$, we let $v_{\mu}(\tau_0,\tau_1)$ be the threshold in terms of the private belief discussed above that Firm $0$ is chosen. That is, choosing Firm  $0$ is uniquely optimal for the consumer if and only if $p(s)>v_{\mu}(\price_{0},\price_{1})$. One can easily see from the above equations that $v_{\mu}(\price_{0},\price_{1})$ has the following form:
\begin{equation}\label{eq:indif_th}
v_\mu(\price_{0},\price_{1})=\begin{cases}
\frac{(1-\mu)(1+\tau_0-\tau_1)}{2\mu-(2\mu-1)(1+\tau_0-\tau_1)}\mbox{ if the market is full,}\\
\frac{(1-\mu)\tau_0}{\mu-(2\mu-1)\tau_0}\mbox{ otherwise.}
\end{cases}
\end{equation}
 Note that $v_\mu(\price_{0},\price_{1})$ is a continuous function of $(\mu,\price_{0},\price_{1}).$

We can therefore suppress the behavior of the consumer, which, under Assumption \ref{asumption}, is determined uniquely for every price vector $(\tau_0,\tau_1)$ and almost every signal realization $s\in S$. Hence we can
 write the expected utility of Firm $0$ in the game $\Gamma(\mu)$ for the price vector $\price$ as follows:
\begin{equation}\label{eq:zero_profit1}
\begin{split}
&\Pi_{0}(\tau_0,\tau_1,\publicBelief)=\\
&\left(\mu \left(1-G_0(v_\mu(\tau_0,\tau_1)\right)+(1-\mu)\left(1-G_1(v_\mu(\tau_0,\tau_1)\right)\right)\tau_0=\\
&\left[1-\left(\publicBelief G_{0}\left(v_{\mu}(\price_{0},\price_{1})\right)+
(1-\publicBelief)G_{1}\left(v_{\mu}(\price_{0},\price_{1})\right)\right)\right]\price_{0}.
\end{split}
\end{equation}
 A similar equation can be derived for $\Pi_{1}(\price,\publicBelief)$, the profit of Firm $1$.

For a strategy profile $\phi=(\phi_{0},\phi_{1},\sigma),$ let $Pr_{\phi,\mu}$ be the probability over $\Omega\times[0,1]^2\times S$; the state, the price vector, and the signal set $S$ are induced by $\phi$, $\mu$ and $F_0,F_1$.

In what follows we make a distinction between two forms of perfect Bayesian Nash equilibria of the game $\Gamma(\mu)$: \emph{a deterrence equilibrium}, where only a single firm sells with positive probability, and \emph{a non-deterrence equilibrium}, where both firms sell with positive probability. That is,
\begin{definition}\label{def:DE}
 \textit{A deterrence equilibrium } in ${\Gamma(\publicBelief)}$ is a Bayesian Nash SPE, $(\phi_{0},\phi_{1},\sigma)$, in which there exists a unique firm $i$ such that $$Pr_{\phi,\mu}(\sigma(\mu,s,\tau)=i) \not = 0.$$
\textit{A non-deterrence equilibrium} is an equilibrium that is not a deterrence equilibrium.
 \end{definition}

The following theorem summarizes the main characteristics of equilibria in the stage game $\Gamma(\mu)$. This characterization is the driving force behind the proof of Theorem \ref{thm:social_learning}.
\begin{theorem}\label{thm:SSG_results}
Let $\mu\in(0,1)$ and
let $(\phi_{0},\phi_{1},\sigma)$ be a Bayesian Nash subgame perfect equilibrium of the game $\Gamma(\publicBelief)$:
\begin{enumerate}
  \item If signals are unbounded, then no firm is deterred.
  \item If signals are bounded and exhibit the vanishing likelihood property, then no firm is deterred.
  \item If signals are bounded and do not exhibit the vanishing likelihood property, then:
  \begin{enumerate}
  	\item If $g_1(\ubar{\alpha})>0,$ then for some sufficiently high  prior, $\bar\mu\in(0,1),$ whenever $\mu>\bar\mu$ Firm $1$ is deterred and Firm $0$ captures the whole market.
  	\item If $g_0(\bar{\alpha})>0,$ then for some sufficiently low  prior, $\ubar\mu\in(0,1),$ whenever $\mu<\ubar\mu$ Firm $0$ is deterred and Firm $1$ captures the whole market.
  \end{enumerate}
\end{enumerate}
\end{theorem}

The proof of Theorem \ref{thm:SSG_results} as well as the complete analysis of this stage game is relegated to Appendices \ref{app:proofs_aux} and \ref{app:thm2_proof}.
\subsection{Proof of Theorem \ref{thm:social_learning}}
We next introduce the formal proof of Theorem \ref{thm:social_learning}, based on Theorem \ref{thm:SSG_results}.
Consider the case where
signals exhibit vanishing likelihood. In this case, by Theorem \ref{thm:SSG_results}, at every time $t$ no firm is deterred in $\Gamma(\mu_t)$. Therefore the decision of the consumer remains informative throughout the play. This implies that asymptotic learning always holds.
By contrast, if  signals do not exhibit vanishing likelihood and $g_1(\bar{\alpha})>0,$ then for $\mu>\bar\mu$ all equilibria of $\Gamma(\mu)$ are deterrence equilibria. This implies that when $\mu_t$ crosses $\bar\mu$ all subsequent consumers buy from Firm $0$ and learning stops.

In order to establish Theorem \ref{thm:social_learning} we start with the following corollary of Lemma \ref{lem:notf1} (the proof can be found at Appendix \ref{app:thm2_proof}).
\begin{corollary}\label{cor:asymptotic_learning}
If signals exhibit vanishing likelihood or if signals are unbounded,
then for every $\varepsilon>0$ there exists some $r>\ubar{\alpha}$ and $\delta'>0$ such that if $\mu\in[\varepsilon,1-\varepsilon]$ and $\phi=(\phi_0,\phi_1)$ is a SPE of $\Gamma(\mu)$, then
$$\mathrm{P}_{\mu,\phi}(v_\mu(\tau_0,\tau_1)\geq r)>\delta'.$$
A similar condition holds for Firm $1.$
\end{corollary}
In words, by Theorem \ref{thm:SSG_results},  if signals exhibit vanishing likelihood or if signals are unbounded, then the probability of a consumer going against the herd is positive. Corollary \ref{cor:asymptotic_learning} argues that this probability cannot be arbitrarily close to zero if the prior is bounded away from the edges.

\begin{proof}[Proof of Theorem \ref{thm:social_learning}]
First we show that if the signal structure $(F_0,F_1,S)$ does not exhibit vanishing likelihood, then the martingale of the public belief must converge to an interior point. Indeed, let us assume without loss of generality that ${g}_1(\ubar{\alpha})>0$. Let $(\bar{\phi}_0,\bar{\phi}_1,\bar{\sigma})$ be a myopic equilibrium.  By Observation \ref{obs:equ}, for $\mathrm{P}_{(\bar{\phi}_0,\bar{\phi}_1,\bar{\sigma})}$ almost every history $h_t\in H_t$ the profile $(\bar{\sigma}(h_t),\bar{\phi}_0(h_t),\bar{\phi}_1(h_t))$ is a SPE of $\Gamma(\mu_t)$. By Theorem \ref{thm:SSG_results},  there exists $\bar\publicBelief$ such that $\forall\publicBelief\in(\bar\publicBelief,1)$ there is a unique Bayesian Nash subgame perfect equilibrium of $\Gamma(\publicBelief)$ in which the consumer almost surely chooses Firm  $0$ (Firm $1$ is deterred by Firm  $0$). This implies that  if $\publicBelief_{t}\in(\bar\publicBelief,1]$, then $\publicBelief_{t+1}=\publicBelief_{t}$ with probability one.
We note that since signals are never fully informative it must be the case that $\mu_t<1$ for all $t$ with probability one.
Therefore, if the vanishing likelihood property does not hold, then asymptotic learning fails.

Next we show that if vanishing likelihood holds, then the public belief martingale converges to a limit belief in which the true state is assigned probability one.
It follows from Observation \ref{obs:equ} that $(\bar{\phi}_0(h_t),\bar{\phi}_1(h_t),\bar{\sigma}(h_t))$ is a SPE of $\Gamma(\mu_t)$ for
$\mathrm{P}_{(\bar{\phi}_0,\bar{\phi}_1,\bar{\sigma})}$ almost every history $h_t\in H_t$. Corollary \ref{cor:asymptotic_learning} implies that if $\mu_t\in[\epsilon,1-\epsilon]$ then
for some $\delta'>0$ and $r>\ubar{\alpha}$ the realized price vector $(\tau_0,\tau_1)$ satisfies $v_{\mu_t}(\tau_0,\tau_1)\geq r$ with probability at least $\delta'$.

Since the distribution $G_{0}(\cdot)$  first-order stochastically dominates  $G_{1}(\cdot)$ (see Lemma \ref{lem:aux} in Appendix \ref{sec:Acem}), under any such price vector $(\tau_0,\tau_1)$ there exists a probability at least $G_0(r)>0$ that the consumer will not buy from Firm $0.$
Note that (again by Lemma \ref{lem:aux}) $$\frac{G_{0}(v_{\mu_t}(\tau_0,\tau_1))}{G_{1}(v_{\mu_t}(\tau_0,\tau_1))}\leq \frac{G_{0}(r)}{G_{1}(r)}=\beta<1.$$
Therefore, it follows from Bayes' rule that with probability at least $G_0(r)\delta'$ the public belief $\mu_{t+1}$ satisfies
\begin{equation}\label{eq:lr_changes}
\frac{\publicBelief_{t+1}}{1-\publicBelief_{t+1}}=\frac{\publicBelief_{t}}{1-\publicBelief_{t}}\frac{G_{0}(v_{\mu_t}(\tau_0,\tau_1))}{G_{1}(v_{\mu_t}(\tau_0,\tau_1))}\leq \frac{\publicBelief_{t}}{1-\publicBelief_{t}}\beta.
\end{equation}
Hence, in particular, if $\mu_t\in[\epsilon,1-\epsilon]$, then there exists a positive constant $\eta>0$ such that, with probability at least $G_0(r)\delta'$, it holds that $|\mu_{t+1}-\mu_t|>\eta$.

By Observation \ref{obs:martingale}, the limit $\mu_\infty=\lim_{t\rightarrow\infty}\mu_t$ exists and by the above argument $\mu_\infty\in\{0,1\}$ with probability $1$. This shows that asymptotic learning holds.
\end{proof}
\section{Social Learning and Farsighted Firms}\label{sec:non-myopic}

In this section we show that our main result carries forward to a setting where the firms are farsighted and maximize a discounted expected revenue stream.
We extend our sequential model to the non-myopic case by defining the \emph{non-myopic sequential consumption game}. In this model, as in the myopic case, each firm sets a price in every time period, except that now each firm tries to maximize its discounted sum of the stream of payoffs.
We still retain the perfection assumption with respect to consumers. Given a strategy profile $(\bar{\phi}_0,\bar{\phi}_1,\bar{\sigma})$, in the repeated game,
denote by $\Pi^\delta_i(\bar{\phi}_0,\bar{\phi}_1,\bar{\sigma})$  the expected payoff to Firm $i$ when the discount factor is $\delta>0$. Namely,
$$\Pi^\delta_i(\bar{\phi}_0,\bar{\phi}_1,\bar{\sigma})=E_{(\bar{\phi}_0,\bar{\phi}_1,\bar{\sigma})}\Big((1-\delta)\sum_{t=1}^\infty\delta^{t-1} \Pi^t_i(\bar{\phi}_0(h_t),\bar{\phi}_1(h_t),\bar{\sigma}(h_t)|h_t)\Big).$$
Define a \emph{Bayesian Nash equilibrium} as follows.
\begin{definition}\label{def:equilibrium}
A strategy profile $(\bar{\phi}_0,\bar{\phi}_1,\bar{\sigma})$ constitutes a \emph{Bayesian Nash equilibrium} if:
\begin{itemize}
\item For every strategy $\bar{\psi}_i$ of Firm $i$,
$$\Pi^\delta_i(\bar{\phi}_i,\bar{\phi}_-i,\bar{\sigma})\geq \Pi^\delta_i(\bar{\psi}_i,\bar{\phi}_{-i},\bar{\sigma}).$$
\item For every time $t$ the following condition holds for almost every history $h_t\in H_t$ that is realized in accordance with $\mathbf{P}_{(\bar{\phi}_0,\bar{\phi}_1,\bar{\sigma})}$, every price vector  $(\tau_0,\tau_1)\in[0,1]^2$, and every decision rule $\sigma\in\mathcal{A}$:
 $$u_t(\tau_0,\tau_1,\bar{\sigma}(h_t)|h_t)\geq u_t(\tau_0,\tau_1,\sigma|h_t).$$
\end{itemize}
\end{definition}
Let $(\bar{\phi}_0,\bar{\phi}_1,\bar{\sigma})$ be a strategy profile, let $h_t\in H_t$, and denote by
$\Pi^\delta_i(\bar{\phi}_0,\bar{\phi}_1,\bar{\sigma}|h_t)$  the continuation payoff to Firm $i$ in the subgame starting at history $h_t\in H_t$.
Note that $\Pi^\delta_i(\bar{\phi}_0,\bar{\phi}_1,\bar{\sigma}|h_t)$ is well defined for almost every history $h_t\in H_t$ that is realized in accordance with $\mathbf{P}_{(\bar{\phi}_0,\bar{\phi}_1,\bar{\sigma})}$.
By Definition \ref{def:equilibrium}, if $(\bar{\phi}_0,\bar{\phi}_1,\bar{\sigma})$ constitutes a Bayesian Nash equilibrium, then
$\bar{\phi}_i$ maximizes the continuation payoff $\Pi^\delta_i(\bar{\phi}_0,\bar{\phi}_1,\bar{\sigma}|h_t)$ of Firm $i$ for almost every history $h_t\in H_t$ that is realized in accordance with $\mathbf{P}_{(\bar{\phi}_0,\bar{\phi}_1,\bar{\sigma})}$.

We can now state our first result for farsighted firms.
\begin{theorem}\label{prop:non-myopic-VL-NDE}
If the signal structure $(F_0,F_1,S)$ exhibits the vanishing likelihood property or if signals are unbounded, then asymptotic learning holds for every discount factor $\delta<1$.%
\footnote{Recall this implies that learning takes place for any prior and any corresponding Bayesian Nash equilibrium.}
\end{theorem}

By Theorem \ref{prop:non-myopic-VL-NDE}, vanishing likelihood is a sufficient condition for asymptotic learning.\footnote{The fact that asymptotic learning holds for unbounded signals carries forward to the farsighted case under a similar proof to that of the myopic case.}

For the converse direction we establish the following weaker statement.
%
\begin{theorem}\label{prop:non-myopic-lmpe-VL-NDE}
If signals are bounded and do not exhibit  vanishing likelihood, then asymptotic learning fails in the following sense: for every prior $\mu_0\in(0,1)$ there exists a Bayesian Nash equilibrium for which learning fails.
\end{theorem}
That is, in order for asymptotic learning  always to hold in a Bayesian Nash equilibrium it is necessary and sufficient for the signal structure to exhibit vanishing likelihood.
We conjecture that, in fact, if vanishing likelihood fails then learning fails in any equilibrium.

\section{Discussion}\label{section:extensions}
We turn to discuss four natural questions that arise from our model and analysis:%
\footnote{We thank anonymous reviewers for prompting these questions.}
\begin{itemize}
\item
Do our conclusions hold when the differentiability assumption on the signal distribution (Assumption \ref{asumption}) is relaxed?
\item
What are the implications of vanishing likelihood in the  monopolistic setting?
\item
What are the social welfare implications of our results?
\item
In cases where asymptotic learning holds, how fast do agents learn and consequently buy the superior product?
\end{itemize}

\subsection{General Signals}\label{sec:gen_sig}

Throughout the analysis we have restricted our attention to signal structures $(F_0,F_1,S)$ that satisfy Assumption \ref{asumption}.
In many applications this assumption fails to hold. In particular, Assumption \ref{asumption} does not hold when the set of signals is countable or finite.
It is therefore important to understand whether  our condition can be stated more generally to capture all signal distributions.

Fortunately, it turns out that such a general condition does exist. Let $(F_0,F_1,S)$ be a general signal distribution and let $G_\omega$ be the CDFs as defined in Definition \ref{def:G_omega}.
Define $g_0,g_1\in[0,\infty]$ as follows:
$$g_0=\liminf_{x\rightarrow\ubar{\alpha}^+}\frac{G_0(x)}{x-\ubar{\alpha}} \text{ and }  g_1=\liminf_{x\rightarrow\bar{\alpha}^-} \frac{1-G_1(x)}{\bar{\alpha}-x}.$$
Obviously, $g_0,g_1$ are both well defined. Note that $g_0$ is defined using  the  limit from the left ($x\rightarrow\ubar{\alpha}^+$) whereas $g_1$ uses the limit from the right ($x\rightarrow\bar{\alpha}^-$).

We can now state the more general condition for vanishing likelihood as follows.
\begin{definition}\label{def:general VL}
The signal structure $(F_0,F_1,S)$ satisfies vanishing likelihood if $g_0=g_1=0$.
\end{definition}

Note that if $(F_0,F_1,S)$ satisfies Assumption \ref{asumption}, then the condition in Definition \ref{def:general VL} coincides with the condition in Definition \ref{def:VL}.
Moreover, note that for finite signal distribution we have $g_0=g_1=\infty$, and thus vanishing likelihood fails.
Our results hold verbatim under the more general definition of vanishing likelihood for when firms are myopic.

We omit the proofs for the general setting but note that the underlying ideas for the proofs are similar while their exposition becomes more cumbersome.
The primary reason for this is that with an arbitrary signal structure the consumer can be indifferent between two options (e.g., indifferent between the two products or between a product and exiting) with positive probability. Therefore, given a price pair, the consumer may have more than one best reply. In addition, it is not necessarily the case that any such best reply induces a two-player game between the firms that possesses an equilibrium. The underlying reason is that the consumer strategy may lead to discontinuity in firms' payoffs as a function of prices. Under Assumption \ref{asumption} the consumer has a unique best reply with probability one and such discontinuity can be ignored.
\footnote{The omitted proofs are available from the authors upon request. Although we have not written a rigorous proof for the case where firms are farsighted, we believe that the results  carry through.}

An additional challenge that results from the aforementioned discontinuity pertains to the mere existence of an equilibrium in  $\Gamma(\mu).$ Absent this equilibrium, our results become vacuous. Fortunately, we can use the result of Reny \cite{Reny1999} to overcome this.

Consider the following specific best-reply consumer strategy:
whenever a consumer is indifferent between buying from one firm  and the outside option he always chooses to buy from the firm.
Whenever a consumer is indifferent between buying from Firm $0$ and Firm $1,$  and his expected utility from purchasing a product is at least zero, he chooses the firm that is a priori favorable.
That is, in this case he chooses Firm $0$ whenever $\mu\geq\frac{1}{2}$ and Firm $1$ whenever $\mu<\frac{1}{2}$. In all other cases he strictly prefers one alternative and therefore chooses this alternative.

With this consumer strategy, game $\Gamma(\mu)$ satisfies Reny's {\em  better-reply secure condition} \cite{Reny1999}  for any $\mu\in[0,1].$ Theorem 3.1 in \cite{Reny1999} thus guarantees the existence of a mixed subgame perfect equilibrium in $\Gamma(\mu).$

\subsection{Monopolistic Market}

A natural question to ask is what are the necessary and sufficient conditions for asymptotic learning when there is a single firm that competes against an outside option.  It turns out that the vanishing likelihood condition plays a crucial role in the monopolistic case as well.

More precisely, consider a monopolistic model with a single firm and a binary state space. In state $\omega=0$ the firm offers a high-quality product (the common value is one) whereas in the other state ($\omega=1$) the common value is zero. At each stage the firm sets a price and a new consumer arrives. The consumer receives a private signal and chooses whether or not to buy the good (the outside option is valued at zero). The  methodology and techniques discussed in this paper may be used to show that in the monopolistic setting,  vanishing likelihood guarantees asymptotic learning and the lack of vanishing likelihood (at both ends) implies that asymptotic learning fails.%
\footnote{Formally, the notion of vanishing likelihood is related to two conditions on the signal structure, that correspond to the two extreme values: the highest and lowest possible signals. Indeed, if one of these conditions fails then learning fails in the duopolistic setting. By contrast, in the monopolistic setting failure of learning is guaranteed only when both conditions are not satisfied.}

\subsection{Social Welfare}

The main goal of this paper is to identify  the informational structure that guarantees learning. As we now turn to show, in our model, learning guarantees that eventually all agents purchase the superior product. This implies that vanishing likelihood is a necessary and sufficient condition for the market to be (asymptotically) efficient.%
\footnote{In some variants of the herding model, such as those studied in by Mueller-Frank \cite{Mueller-frank2012,Mueller-frank2016}, learning does not entail market efficiency.}

\begin{corollary}\label{lem:AL_consumer_buys_superior}
	Let $(\overline{\sigma},\overline{\price}_0,\overline{\price}_1)$ be a myopic Bayesian Nash  equilibrium. If asymptotic learning holds, then conditional on state $\omega\in\Omega$,
	$$\lim_{t\rightarrow\infty} \mathbf{P}_{(\overline{\sigma},\overline{\price}_0,\overline{\price}_1)}(\{\sigma_{t}(\mu_t,s,\overline\tau(\mu_t))=\omega\}|\omega)=1.$$
\end{corollary}
The corollary follows from the proof of  our main theorem and its proof is relegated to Appendix D.


\subsection{Rate of Convergence}
A natural question to ask is, in cases where asymptotic learning holds, how fast does the market reach a situation where the consumer buys the superior product with high probability. In the herding model, this question has only recently received  attention (see, e.g.,  Rosenberg and Vieille \cite{Rosenberg2017} and Hann-Caruthers et al. \cite{Hann-Caruthers2018}). In particular, Rosenberg and Vieille \cite{Rosenberg2017} demonstrate a sharp negative result showing that for ``reasonable''    
unbounded signals, the expectation of the first correct decision by a consumer, is infinite. In our setting we conjecture that under vanishing likelihood breaking bad herds is faster and thus the rate of convergence is higher.
We leave this interesting question for future research. 
\section{Case Studies }\label{sec:examples}

An intriguing application of our theoretical results relates the outcome of the evolutionary process of Schumpeterian growth. One particular interest (and concern) is a corollary of our main theorem. Our main theorem  asserts that whenever one firm is sufficiently a-priori advantageous (i.e. the public belief from the outset is that it has the superior product with high probability), it  can use predatory pricing to obstruct the entry of entrants yielding  innovative substitute products and technologies. In fact, whenever the demand side does not exhibit the vanishing likelihood property, predatory pricing will eventually become the incumbent's optimal strategy. 
Recent technological advances had led to an increase in the number and complexity of proposed innovations \cite{OECD2015}. The rise of the Internet increased the exposure to new products and the pace at which adoption decisions are made. Social networks inform potential consumers over the action taken by their predecessors, therefore we suspect that such an advantage is quite typical.}

We illustrate the outcome of our model using two case studies. In one case, which fits the vanishing likelihood requirements, the incumbent indeed priced out a substitute product, considered by many in hindsight as the superior product. In the other, which fits the opposite scenario, the incumbent did not engage in predatory pricing and indeed the new product, which is obviously the superior one, prevailed.

Given our informal discussion and the underlying assumptions we put forward for these two cases, in particular with respect to the prevalence of (or lack thereof) vanishing likelihood, we recommend to read them with a grain of salt. In addition, even if one agrees with our assumptions one can suggest alternative explanations to account for the different incumbent reaction in both cases (some of which we mention in footnotes).


\color{black}

\noindent {\bf A tale of vanishing likelihood:}
In the mid-1990s, with around 10 percent market share, Barnes and Noble (BN) was a clear leader in the US book-selling market \cite{Ghemawat2004,Greco2013book}. BN rose to power by perfecting the shopping experience for buyers and through aggressive discounts, of $20-30$ percent of cover prices \cite{Grimes1995}.

In 1995 BN faced a new type of competition: Amazon's online retail. In hindsight this is a classic example of ``creative destruction"; however, back then the success of online retail was ambiguous. In the late 90s the Internet was new and was definitely not conceived as a retail shopping channel \cite{Cleland1995}. Uncertainty about the security of online payments was immense \cite{Smith1998,Crisp1995} and many considered instant gratification, absent from online shopping, as central to the shopping experience.

To overcome these challenges, Amazon offered 40 percent discounts of cover prices, while BN maintained its 20--30 percent discount. Early adopters of the new technology started shopping online, followed by more conservative shoppers. Fourteen years later Amazon was a clear market leader while BN closed most of its shops. With a $16.5\%$ market share and over $\$400$B in trade in the 2016 US market \cite{NationalRetailFederation2016}, online retail shopping is a proven superior innovation.

We note that BN, with its 28 percent operational profitability at the time \cite{Brooker2000}, could have lowered prices and driven Amazon out of the market. Was it rational to maintain its high profit or should BN have offered additional discounts given the information available at the time? We shall revisit this issue after the next case study.

\noindent {\bf A tale of non-vanishing likelihood:}
In the game console market, the quality of a product is determined by a variety of measurable determinants (e.g., graphic capabilities, CPU speed, RAM) and unmeasurable ones (e.g., design, gameplay, and game titles). Back in the late 1990s, game consoles offered no connected gaming and so network effects and externalities were less significant in console adoption.\footnote{One may still claim that network effect played an important role in ``Dreamcast" failure.  One of the authors, a enthusiastic `gamer' at the time Sega launched the Dreamcast,  contend this network effect for the following reasons: (1) ``Sega'' was a veteran game manufacturer, which held the intellectual property rights for various strong titles such as: ``Golden axe'', ``Zelda'', ``Sonic'' etc. These titles where not available to ``Sony". ;and (2) Sony's next console, the ``Playstation II'', did not reach the market until the following year and was known to offer no backward compatibility. Therefore it is unclear which of the firms enjoyed a stronger the network effect.}

The  Dreamcast game console, developed at the end of the 1990s, was  Sega's  second attempt to restore its place as an industry leader.%
\footnote{Sega's previous console, the ``Saturn," had failed miserably due to its ``awful gameplay and inferior design" \cite{TheVideoGameCritic}.}
The market leader at the time was Sony with its PlayStation $I$ console.
In contrast with BN's  strategy of ignoring Amazon, Sony preempted the launch of Dreamcast by offering a $30$ percent discount on its own console a month before the release of Dreamcast and one year before releasing its own next-gen product  \cite{Snodgrass1999,Wikia2017}.
As a result, Dreamcast did not manage to penetrate the market and Sega stopped making and selling this console less than two years after its introduction.%
\footnote{Indeed, Tadashi Takezaki, a former executive at Sega, points out two main factors related to the failure of Dreamcast: (1) Consumer skepticism about Sega's abilities to produce a viable product following the aforementioned ``Saturn" fiasco and (2) Sony's aggressive pricing \cite{Gifford2013}.}

In retrospect, Dreamcast was acknowledged as the superior technology and some even consider it one of the best consoles ever developed. This reasonably priced and developer-friendly technology contained many futuristic features such as ``network gameplay" and 64-bit high resolution graphics.\footnote{Even today, over a decade and a half after its initial release, second-hand consoles are being traded on eBay and new Dreamcast titles are being released  \cite{Brooks2016}, which serves as a testimony to the console's technical superiority.}
Some of these novel features did not reach the market until more than half a decade later.%
\footnote{Online multiplayer interface was introduced by  PlayStation III in 2006, the same year a control system with motion detection first appeared in the Nintendo Wii.}
In Schumpeter's framework, growth had been delayed.

\noindent {\bf Vanishing vs. non-vanishing likelihood:}
Why did the incumbent firm price aggressively, thus interfering with the evolutionary process,  in one case and not in the other? What is the primary distinction that explains this? One explanation that resonates with our model has to do with the different nature of the demand side of both markets. Buyers of game consoles are highly engaged and actively seek information regarding new releases and so we witness more diverse opinions by buyers and a substantial proportion of nonconformism. By contrast, book buyers, who do show interest and are engaged when it comes to the decision of which title to buy, are less inquisitive when it comes to the actual shopping experience. Only a marginal proportion would not conform to traditional shopping practices.
Thus the marginal proportion of nonconformists, captured by our notion of vanishing likelihood, is significant in the console market and is insignificant in the book market.\footnote{An alternative explanation may be that demand for most products is likely to be multi-dimensional. We argue that in these two cases studies, pricing plays a major part in the consumer's decision. In the Amazon case, consumer price sensitivity played a major part in the raise of the book-superstores. It is therefore only logical to assume that it was also meaningful in their demise. In the gaming console example, the majority of potential consumers  at the time were not the gamers themselves, but their parents, which were more price sensitive (see \cite{Gifford2013}).} For Sony it was rational to reduce prices in order to maintain the nonconformist market share. For BN, on the other hand, the size of the nonconformist market share coupled with the ambiguity of the success of the new online shopping paradigm was insufficient to forego profits.%
\footnote{One could suggest that the nonconformism of the game console market is captured by the notion of unbounded signals whereas the book market exhibits bounded signals. Adopting this approach must lead to the conclusion that creative destruction would be observed in the console market and not in the book market, contrary to the evidence.}

\section{Summary}\label{section:discussion}
In the classic models of social learning, the consumers' utility from each alternative is fixed. Thus, when signals are bounded,  there is always a positive probability that the inferior product will prevail \cite{Banerjee1992,Bikhchandani1992}.  However, when signals are unbounded there are always nonconformist consumers who go against the herd and purchase the product most others won't (see Smith and S\o rensen \cite{Smith2012}). These nonconformists are instrumental for the aforementioned information aggregation when the consumers' choice is between products with fixed prices.
However, our setting involves strategic pricing that alters these results. As each of the firms can now lower its price and attract consumers, even when the prior belief is biased against it, or price out its competitor if the prior belief is in its favor.  We ask two questions: When are these pricing strategies optimal? And what implications do these strategies have for the manner in which markets aggregate information?  We find that the proportion of nonconformist consumers plays a significant role in answering these questions. An intuitive extension of the model with fixed prices suggests that more nonconformists implies more social learning. However, our main finding is the exact opposite: social learning occurs only when the number of nonconformist consumers is small. This is the condition we refer to as ``vanishing likelihood."

We study the conditions under which markets in which firms are engaged in a pricing competition enable or hinder social learning. We do so by introducing a simple setting of duopolistic pricing competition.  We first study a simplified model where firms are myopic and prove that in this setting, when signals are bounded, social learning occurs if and only if the signal distributions exhibit the vanishing likelihood property.  We then extend these results to a version of the model with forward-looking firms that maximize their expected discounted future revenue stream.

The rationale behind this counterintuitive result is uncovered when analyzing the firms' incentives in the stage game. As society learns, one of the firms, say Firm  $0$, emerges as the better one. At that stage the new consumer, prior to receiving a signal, assigns a high probability to Firm  $0$ having the superior product. In other words, the stage game begins with a biased prior toward Firm  $0$, which now wants to exploit this near-monopolistic status and set a high price. The only reason not to do so is when the next consumer is very likely to receive a strong signal that Firm $1$ is superior and consequently does not conform with his predecessors. This argument can be ignored by Firm  $0$ when the probability of this event is low enough, which is exactly captured by our notion of ``vanishing likelihood.'' Therefore, when signals exhibit vanishing likelihood, the popular firm ignores nonconformist consumers and foregoes this market share by setting prices high. Firm $1$ sets prices low and wins over the consumer in the rare event that he is a nonconformist, thus breaking the herd phenomenon. Notice that no matter how small the probability of this is in the stage game, when we go back to the repeated game it eventually happens with probability one.

While our major contribution is to the literature on social learning,  the vanishing likelihood property and its effect on firms' strategic behavior has interesting implications for market behavior. In particular for market entry and the adoption of new technologies.  Previous work which studied such questions assumed that incumbents have, either an informational advantage (see \cite{Bagwell2007}) or the ``first move" advantage and can preempt entry by increasing capacity or investing in R\&D  (\cite{Acemoglu2015,Barrachina2014}) or both (see \cite{Milgrom1982,Milgrom1982a}). To the best of our knowledge, our stage game is the first example of predatory pricing behavior, when both incumbent and entrant act simultaneously, and no firm has an informational advantage. We study  these aspects in a companion paper (Arieli et al. \cite{arieli2016predatory}).

In addition to the theoretical contribution of the paper, we believe that the
	intuition uncovered in our model merits further discussion and may pave
	the way to subsequent work in related fields. For example, a stylized version
	of the model may be used to examine questions of firm pricing dynamics,
	or the interim social cost associated with predatory pricing in experience
	good markets.
	Note that the market failure we uncover, can be easily rectified by employing a minimum price policy (Similarly to the policy employed to enable
	entry to the Israeli communications market \cite{Scheer2014}). The costs and benefits of
	such policy are left for follow-up work.

Another potential direction may be an empirical examination of technology adoption success and failures.  While it may  be difficult to directly ascertain whether  a distribution satisfies the ``vanishing likelihood property", it is possible to examine whether a large proportion of non-conformists aids or obstructs the process of innovation adoptions (similarly to \cite{Klier2016}). 

\newpage
\bibliographystyle{plain}
\bibliography{herding}{}
\appendix

\section{Proofs of the Stage Game}\label{app:proofs_aux}

\subsection{Equilibrium Analysis of $\Gamma(\mu)$}
We start with some preliminary results concerning equilibrium behavior in the game $\Gamma(\mu)$.
Throughout we use the notation $\phi_i\in\Delta([0,1])$ for a mixed strategy of Firm $i$ and $\tau_i$ for a pure strategy.

For $\mu\in[0,1]$ we use the following shorthand: $G_\mu(x)=\mu G_0(x)+(1-\mu) G_1(x).$
%
It follows by equation \eqref{eq:zero_profit1} that whenever the consumer's strategy $\sigma$ obeys equation \eqref{eq:consumer_condition}, the expected utility of Firm $0$ in the game $\Gamma(\mu)$, $\Pi_0(\tau_0,\tau_1,\sigma),$ can be written as follows:
\begin{equation}\label{eq:zero_profit}
\begin{split}
&\Pi_0(\tau_0,\tau_1,\sigma)=(1-G_\mu(v_\mu(\tau_0,\tau_1)))\tau_0.
\end{split}
\end{equation}

For a mixed strategy profile $(\phi_0,\phi_1),$ let $\phi\in\Delta([0,1]\times [0,1])$ be the price probability distribution  $(\phi_0,\phi_1)$ induced over $[0,1]\times [0,1]$.
By equation \eqref{eq:zero_profit},  Firm $0$'s payoff from the mixed strategy profile $(\phi_0,\phi_1)$ can be written as follows:
\begin{equation}\label{eq:zero_profit_mix}
\begin{split}
&\Pi_0(\phi_0,\phi_1,\sigma)=\Pi_{0}(\phi_0,\phi_1)=\\
&\int\big(\mu(1-G_0(v_\mu(\tau_0,\tau_1)+(1-\mu)(1-G_1(v_\mu(\tau_0,\tau_1))\big)\tau_0d\phi(\tau_0,\tau_1),
\end{split}
\end{equation}
where $v_\mu(\cdot,\cdot)$ is defined as equation \eqref{eq:indif_th}.

The next lemma provides an alternative way to write $v_\mu(\price_{0},\price_{1})$ and its derivative.
This will turn out to be useful in the sequel.
Consider the following function ${\bar v}_{\mu}:[0,1]^2\rightarrow\mathbb{R}:$
\begin{equation}\label{eq:treshold_function}
{\bar v}_{\mu}(\price_{0},\price_{1})\equiv
\begin{cases}
 \LLR[\frac{1+\price_{0}-\price{1}}{2}]- \LLR[\publicBelief]]&\mbox{ if the market is full,}\\
 \LLR[\price_{0}] - \LLR[\publicBelief]&\mbox{ if the market is not full.}
\end{cases}
\end{equation}
\begin{lemma}\label{lem:llr}
It holds that
\begin{equation}
\Diff{{\bar v}_{\mu}(\price)}{\price_{0}}|_{\eqPrice(\publicBelief)}=\begin{cases}
\label{eq:v_diff_gen1}\frac{2}{1-(\tau_{0}-\tau_{1})^2}\mbox{ if the market is full,}\\
\frac{1}{\tau_{0}(1-\tau_{0})}\mbox{ if the market is not full}
\end{cases}
\end{equation}
and
\begin{equation}\label{eq:llr}
\Diff{{ v}_{\mu}(\price)}{\price_{0}}=\frac{e^{\bar{v}_\mu(\tau_0,\tau_1)}}{(1+e^{\bar{v}_\mu(\tau)})^2}\frac{\partial {\bar v}_\mu(\tau_0,\tau_1)}{\partial \tau_0}.
\end{equation}
\end{lemma}
\begin{proof}
The proof makes standard use of the log-likelihood ratio transformation (see, e.g., Smith and S\o rensen \cite{Smith2012}, Herrera and H\o rner \cite{Herrera2013}, and Duffie et al. \cite{Duffie2014}). The log-likelihood ratio of a belief $p\in[0,1]$ is given by
$\log(\frac{p}{1-p})$. In particular, the log likelihood  ratio of the posterior belief is
\begin{equation}\label{eq:adding logs}
\LLR[\privateBelief_{\publicBelief}(\signal)]=\LLR[\publicBelief]+\LLR[\privateBelief(\signal)].
\end{equation}
It follows from equation \eqref{eq:indif_th} that a consumer with private belief $\privateBelief_{\publicBelief}(\signal)$ prefers  Firm $0$ if and only if $$\LLR[\privateBelief_{\publicBelief}(\signal)]\geq\log(\frac{v_\mu(\tau)}{1-v_\mu(\tau_0,\tau_1)})={\bar v}_\mu(\tau_0,\tau_1).$$
 Equation \eqref{eq:llr} then follows directly from the fact that $v_\mu(\tau)=\frac{e^{\bar{v}_\mu(\tau)}}{1+e^{\bar{v}_\mu(\tau)}}$.
\end{proof}
A simple lemma that turns out be  useful in our analysis is the following:
\begin{observation}\label{lem:prices}
Let $\mu\in [0,1]$ and let $(\phi_0,\phi_1)$ be a SPE of $\Gamma(\mu)$.
The following properties hold:
$$\phi_0([2\lBound-1,1])=1 \text{ and }\phi_1([1-2\uBound,1])=1.$$
\end{observation}
\begin{proof}
We prove the observation for Firm $0$. Note that if $\lBound\leq\frac{1}{2}$ we have nothing to prove.
Assume that $\lBound>\frac{1}{2};$ in this case we have that if $\tau_0=2\lBound-1$, then the consumer will buy from Firm $0$ with probability one for almost every signal realization $s$ and every price $\tau_1\geq 0$ of Firm $1$. To see this, note that $p_\mu(s)> \lBound$ for almost every signal $s\in S$. Therefore,
$$p_\mu(s)-(2\lBound-1)>1-p_\mu(s).$$
This shows that for a price $\tau_0=2\lBound-1$ the consumer buys from Firm $0$ with probability one even for $\tau_1=0$.
In particular, under any price $\tau_0\leq 2\lBound-1$ the expected profit of Firm $0$ is $\tau_0$.
Therefore, if  $\lBound>\frac{1}{2}$ the price $2\lBound-1$ strictly dominates all prices $\tau_0<2\lBound-1$ for Firm $0$.
\end{proof}
\subsection{Properties of Deterrence Equilibria}\label{section:propde}
A key property of a deterrence equilibrium is given  in the following lemma.
\begin{lemma}\label{lem: alpha greater than half}
Let $(\phi_0,\phi_1)$ be a deterrence equilibrium  in the game $\Gamma(\publicBelief)$. If Firm  $0$ controls the market, then $\lBound\geq\frac{1}{2}$. Symmetrically, if Firm $1$ controls the market then $\uBound\leq\frac{1}{2}$.
\end{lemma}
In words, if Firm $i$ is driven out of the market (in the sense that the consumer surely does not buy from it) it  must be the case that the consumer's posterior belief assigns  a probability of at most $\frac{1}{2}$ that Firm $i$ is the superior firm.

\begin{proof}
Assume to the contrary that $\lBound<\frac{1}{2}$ and that $(\phi_0,\phi_1)$ is a deterrence equilibrium in which Firm $1$ is deterred. In this case $\Pi_{1}(\phi_0,\phi_1)=0.$
Consider a deviation of Firm $1$ to the pure strategy $\tau_1=\frac{1-2\lBound}{2}>0$.  By equation \eqref{eq:consumer_condition} we can conclude that any consumer whose signal falls in the set $\{s\in S| \posterior\in [\lBound,\lBound+\frac{\lBound}{2}+\frac{1}{4})\}$ will choose Firm $1$ with probability one for any equilibrium strategy $\phi_0$ for Firm $0$.

 Note that the set $\{s\in S| \posterior\in [\lBound,\lBound+\varepsilon)\}$ has positive probability for every $\varepsilon>0$ and in particular for
$\varepsilon = \lBound+\frac{\lBound}{2}+\frac{1}{4}$. Therefore this deviation entails a positive expected utility for Firm $1$ and hence a profitable deviation, thus contradicting the equilibrium assumption.
\end{proof}

%
%

Given a deterrence equilibrium, there is always a unique price for the firm that controls the market. The price of the firm that controls the market is determined such that its least favorable consumer becomes indifferent between buying from the dominant firm or receiving the other firm's product for free.
Thus, under deterrence equilibrium the consumer buys from the dominant firm with probability one.

\begin{proposition}\label{prop: DE price}
If $(\phi_0,\phi_1)$ is a deterrence equilibrium in game $\Gamma(\publicBelief)$, then
\begin{align*}
\phi_0=2\lBound-1\text{ and }\Pi_0(\phi_0,\phi_1)=2\lBound-1 &\mbox{ if }\lBound\geq\frac{1}{2}\\
\phi_1=1-2\uBound\text{ and }\Pi_0(\phi_0,\phi_1)=1-2\uBound &\mbox{ if }\uBound\leq \frac{1}{2}.\\
\end{align*}
\end{proposition}

\begin{proof}
Let us assume without loss of generality that  Firm $1$ is deterred and so $\lBound\geq\frac{1}{2}$ (By Lemma \ref{lem: alpha greater than half}).
It follows from Observation \ref{lem:prices} that $\phi_{0}([2\lBound-1,1])=1$.
Assume by way of contradiction that $\phi_{0}[2\lBound-1+\delta,1]>0$ for some positive $\delta>0$ and consider the price $\tilde{\tau}_{1}=\frac{\delta}{2}$
for  Firm $1$ (the deterred firm). In this case, for any realized $\tau_0\in [2\lBound-1+\delta,1],$ any consumer with a private signal $s$ such that $\posterior\in [\lBound,\lBound+\frac{\delta}{4}]$, an event whose probability is positive, will buy from Firm  $1$, which, in turn, will have a positive utility. In the deterrence equilibrium Firm $1$'s utility is obviously zero and hence the price $\tilde{\tau}_{1}=\frac{\delta}{2}$ constitutes a profitable deviation, thus contradicting the equilibrium assumption. Therefore $\phi_{0}[2\lBound-1+\delta,1]=0$ for any $\delta>0$. Hence Firm $0$ plays $\tau_{0}=2\lBound-1$ with probability one, as claimed.
\end{proof}

By Lemma \ref{lem: alpha greater than half}, the condition  $\lBound\ge\frac{1}{2}$ is necessary in order for a deterrence equilibrium (in which Firm $1$ is deterred) to exist. We now turn to study the implications of this condition.

\begin{lemma}\label{lem:NDE higher price}
If $(\phi_{0},\phi_{1})$ is a non-deterrence Bayesian Nash SPE of $\Gamma(\publicBelief)$, then the following conditions hold: $\phi_0((2\lBound-1,1))>0$, $\Pi_{0}(\phi_{0},\phi_{1},\sigma)\geq 2\lBound-1,$ and $\Pi_{1}(\phi_{0},\phi_{1})>0$. Symmetrically for Firm $1$, $\phi_1((1-2\uBound,1))>0$,  $\Pi_{1}(\phi_{0},\phi_{1},\sigma)\geq 1-2\uBound,$ and $\Pi_{0}(\phi_{0},\phi_{1})>0.$
\end{lemma}
\begin{proof}
We prove the first part of the lemma.
Lemma \ref{lem:prices} implies that $\phi_0([2\lBound-1,1])=1.$
We further note that if $(\phi_{0},\phi_{1})$ is a SPE profile for which $\phi_0$ is the Dirac measure on $2\lBound-1$, then $\Pi_{0}(\phi_{0},\phi_{1})=2\lBound-1,$ which means that the consumer buys from Firm $0$ with probability $1$.
Hence such an equilibrium must be a deterrence equilibrium. Therefore, it must hold that $\phi_0((2\lBound-1,1))>0$.

The fact that $\Pi_{1}(\phi_{0},\phi_{1})>0$ follows since, as in the proof of Proposition \ref{prop: DE price}, if $\phi_0((2\lBound-1,1))>0$, then Firm $1$ can guarantee a positive payoff against $\phi_0$.
\end{proof}

\section{Proof of Theorem \ref{thm:SSG_results}}\label{app:thm2_proof}

\subsection*{Unbounded signals}
We begin the proof of Theorem \ref{thm:SSG_results}, by studying the case of unbounded signals, i.e., where $\ubar{\alpha}=0$ and $\bar{\alpha}=1.$
The following corollary shows that whenever signals are unbounded there can not be a deterrence equilibrium. In fact all equilibria are non-deterrence equilibria.

\begin{corollary}\label{cor:unbounded NDE}
If signals are unbounded then there are no deterrence equilibria in $\Gamma(\mu)$.
\end{corollary}
\begin{proof}
Since $\bar{\limitParam}=0$ and $\ubar{\limitParam}=1$ it follows that  $\lBound=0$
and $\uBound=1$.  The proof now follows from Lemma  \ref{lem: alpha greater than half}.
\end{proof}

\subsection*{Bounded signals with vanishing likelihood}

We now consider the case where signals are bounded, i.e., $\ubar{\alpha},\bar{\alpha}\in(0,1),$ and signals exhibit the vanishing likelihood property, i.e., $g_1(\ubar{\alpha})=0$. In the following lemma we show that the vanishing likelihood property also yields that $g_0(\ubar{\alpha})=0.$

\begin{lemma}\label{lem:g_0=0_in_VL}
	If the signal structure $(F_0,F_1,S)$ exhibits vanishing likelihood then $g_0(\ubar{\alpha})=0.$
\end{lemma}	
\begin{proof}
	Since vanishing likelihood holds we have that $g_1(\ubar{\alpha})=0.$  Assume to the contrary that $g_1(\ubar{\alpha})< g_0(\ubar{\alpha}).$ Since $F_0,F_1$ are absolutely mutually continuous, there exists $\varepsilon$ such that $\int_{\ubar{\alpha}}^{\ubar{\alpha}+\varepsilon}g_0(s) ds > \int_{\ubar{\alpha}}^{\ubar{\alpha}+\varepsilon}g_1(s)ds \Rightarrow G_0(\ubar{\alpha}+\varepsilon)> G_1(\ubar{\alpha}+\varepsilon)$. This stands in contradiction to Lemma \ref{lem:aux} in Appendix \ref{sec:Acem}, which shows that $G_0$ first order stochastically dominates $G_1$.
\end{proof}

The second part of Theorem \ref{thm:SSG_results} is proved in the following proposition.
\begin{proposition}\label{prop:no DE when g=0}
If the signal structure $(F_0,F_1,S)$ exhibits vanishing likelihood, then for every $\mu\in(0,1)$ there is no deterrence equilibrium in $\Gamma(\publicBelief)$.
\end{proposition}
\begin{proof}
Without loss of generality assume that $\mu\in(\frac{1}{2},1)$ and  assume to the contrary that there exists a deterrence equilibrium in $\Gamma(\mu).$ By Lemma \ref{lem: alpha greater than half}, the only possible deterrence equilibrium is one in which Firm $1$ is deterred and by Proposition \ref{prop: DE price} it must take the form of  $(2\ubar{\alpha}_\mu-1,\phi_1).$ Therefore, $\Pi_0(2\lBound-1,\phi_1)=2\lBound-1.$
We first claim that it is sufficient to show that
\begin{equation}\label{eq:profit_de}
\Pi_0(2\ubar{\alpha}_\mu-1+\varepsilon,0)-\Pi_0(2\ubar{\alpha}_\mu-1,0)>0
\end{equation}
for some $\varepsilon>0.$
 To see this, note that
$\Pi_0(2\ubar{\alpha}_\mu-1,\phi_1)=2\ubar{\alpha}_\mu-1$ for any mixed strategy $\phi_1$ of Firm $1$. In addition, for any fixed price $\tau_0$ of Firm $0$ the payoff
$\Pi_0(\tau_0,\tau_1)$ is (weakly) decreasing in $\tau_1$. Therefore, the inequality in \eqref{eq:profit_de} implies that $\Pi_0(2\ubar{\alpha}_\mu-1+\varepsilon,\phi_1)>2\ubar{\alpha}_\mu-1$ for any mixed strategy $\phi_1$ of Firm $1$. Therefore, if $2\lBound-1+\varepsilon$ yields a profitable deviation to Firm $0$ against price $\tau_1=0$ it also yields a profitable deviation with respect to any strategy $\phi_1$ of Firm $1.$

To establish equation \eqref{eq:profit_de} note that
\begin{equation}\label{eq:direv zero1}
\begin{split}
&\Diff{\Pi_{0}(\price_{0},0)}{\price_{0}}|_{\price_0=2\lBound-1}=\\
&1-(2\lBound-1)\left(\Diff{v_{\mu}(\price_{0},0)}{\price_{0}}|_{2\lBound-1}\right)
\left(\publicBelief g_{0}(\ubar{\alpha})+(1-\publicBelief)g_{1}(\ubar{\alpha})\right).
\end{split}
\end{equation}

Since vanishing likelihood holds, by Lemma \ref{lem:g_0=0_in_VL} we have that $g_1(\ubar{\alpha})=g_0(\ubar{\alpha})=0.$
Therefore equation \eqref{eq:direv zero1} implies that
$$\Diff{\Pi_{0}(\price_{0},0)}{\price_{0}}|_{\price_0=2\lBound-1}=1.$$
Hence $\Pi_0(2\ubar{\alpha}_\mu-1+\varepsilon,0)-\Pi_0(2\ubar{\alpha}_\mu-1,0)>0$ for all sufficiently small $\varepsilon>0$.

\end{proof}
If a deterrence equilibrium does not exist then at each stage of our sequential setting the actual action of the consumer will give us additional information and the public belief will shift. Intuitively, this drives the learning result. However, it turns out that this is not enough. Herrera and H\o rner \cite{Herrera2013} show that in the herding model the fact that $\mu_t\neq \mu_{t+1}$ with probability one does not imply that asymptotic learning holds.
In order to establish asymptotic learning the following stronger result is required.
\begin{lemma}\label{lem:notf1}
If the signal structure $(F_0,F_1,S)$ exhibits the vanishing likelihood condition or signals are unbounded, then
for every $\varepsilon>0$ there exists $\delta>0$ such that if $\mu\in[\epsilon,1-\epsilon]$ and $\phi=(\phi_0,\phi_1,\sigma)$ is a SPE of $\Gamma(\mu),$
then $\mathrm{P}_{\mu,\phi}(\sigma(\tau_0,\tau_1,s)=a)\leq 1-\delta$ for any Firm $i=0,1$.
\end{lemma}
Note that Lemma \ref{lem:notf1} subsumes Proposition \ref{prop:no DE when g=0} in that under vanishing likelihood,  a deterrence equilibrium does not exist. This applies that there exists an upper bound on the probability that the consumer buys from any given firm. This in turn implies that if $\mu$ is bounded away from zero and one, then the distance between $\mu$ and the posterior probability, conditional on the action of the current consumer, is bounded away from zero.
\begin{proof}
We prove the lemma under the assumption that the signal structure $(F_0,F_1,S)$ exhibits vanishing likelihood. The proof for the unbounded case is similar and therefore omitted.

Assume by way of contradiction that there exists an $\epsilon>0$ and a sequence of SPE $\phi^k=(\phi^k_0,\phi^k_1,\sigma_k)$ of $\Gamma(\mu_k)$ such that
$\mu_k\leq 1-\epsilon$ and $$\lim_{k\rightarrow\infty}\mathrm{P}_{\mu_k,\phi^k}(\sigma_k(\tau_0,\tau_1,s)=0)=1.$$
In words, as $k$ increases, the prior approaches $1.$ We show that this entails that the probability of the consumer buying from Firm $0$ approaches $1$ as well.

We can clearly assume (possibly by considering subsequences) that the sequence $\{\Pi_\omega(\phi^k_0,\phi^k_1,\sigma_k)\}_{k=1}^\infty$ converges. Similarly, we can assume that $\{(\phi^k_0,\phi^k_1,\mu_k)\}_{k=1}^\infty$ converges to some $(\phi_0,\phi_1,\mu).$\footnote{The convergence of $\phi^k_\omega$ is assumed with respect to the weak topology.}
As $\lim_{k\rightarrow\infty}\mathrm{P}_{\mu_k,\phi^k}(\sigma_k(a=0))=1$, the limit profit of Firm $1$ shrinks to zero: $$\lim_{k\rightarrow\infty}\Pi_1(\phi^k_0,\phi^k_1,\sigma_k)=0.$$

It follows that the limit price of Firm $0,$ $\lim_{k\rightarrow\infty}\phi^k_0,$  is the pure deterrence price $2\ubar{\alpha}_{\mu}-1$. To see this consider the case where  $\phi_0((2\ubar{\alpha}_{\mu}-1+\eta,1))>0$ for some $\eta>0.$ In this case Firm $1$ play $\tau'_1=\frac{\eta}{2}$ in $\Gamma(\mu)$ and guarantee a positive profit against $\phi_0$. Since in this case the consumers for which
$\posterior\in [\ubar{\alpha}_{\mu},\ubar{\alpha}_{\mu}+\frac{\delta}{4})$ will strictly prefer to buy from Firm $1.$  Thus, at the constant price of $\tau'_1=\frac{\eta}{2}$ Firm $1$ can guarantee a positive expected payoff that is bounded away from zero, for all sufficiently large $k$.

Consider the game $\Gamma(\mu_k)$ and the strategy profile $(\phi^k_0,\phi^k_1)=(2\ubar{\alpha}_{\mu_k}-1,
\phi^k_1).$  A standard continuity consideration implies that since $\phi^k=(\phi^k_0,\phi^k_1)$ is a SPE of $\Gamma(\mu_k)$ and $\lim_{k\rightarrow\infty}(\phi^k_0,\phi^k_1,\mu_k)=(\phi_0,\phi_1,\mu),$ it holds that $(\phi_0,\phi_1)$ is a SPE of $\Gamma(\mu).$  Therefore $(\phi_0,\phi_1)=(2\ubar{\alpha}_{\mu}-1,\phi_1).$ Under the price $2\ubar{\alpha}_{\mu}-1,$ the consumer buys from Firm $0$ with probability one.

 This yields that $(\phi_0,\phi_1)$ is a deterrence equilibrium of $\Gamma(\mu),$ which stands in contradiction to Proposition \ref{prop: DE price}.
\end{proof}

We get the following corollary of Lemma \ref{lem:notf1}.
\begin{cor1}
If signals exhibit vanishing likelihood or if signals are unbounded, then for every $\varepsilon>0$ there exists some $r>\ubar{\alpha}$ and $\delta'>0$ such that if $\mu\in[\varepsilon,1-\varepsilon]$ and $\phi=(\phi_0,\phi_1)$ is a SPE of $\Gamma(\mu)$, then
$$\mathrm{P}_{\mu,\phi}(v_\mu(\tau_0,\tau_1)\geq r)>\delta'.$$
A similar condition holds for Firm $1.$
\end{cor1}
\subsection*{Bounded signals without vanishing likelihood}
The following lemma shows that the consumer's threshold signal approaches the lower bound $\ubar{\alpha}$ as $\mu$ approaches $1$ in every SPE.
\begin{lemma}\label{cor: along the path}
Let $\{\mu_k\}_{k=1}^{\infty}\subseteq(0,1)$ be a sequence of priors such that $\lim_{k\rightarrow \infty}\mu_{k}=1$. Let  $\phi^k=(\phi^k_{0},\phi^k_{1},\sigma_k)$ be a SPE  for the game $\Gamma(\mu_k)$. Then the following holds for every $\epsilon>0:$
$$lim_{k\rightarrow\infty} \mathrm{P}_{\mu_k,\phi^k}(v_{\mu_k}(\tau_{0},\tau_{1})\in [\ubar{\alpha},\ubar{\alpha}+\epsilon])=1.$$
\end{lemma}

\begin{proof}
Assume by way of contradiction that there exists some $\epsilon_0>0$ and $\delta>0$ for which the following holds (possibly considering a subsequence): $$lim_{k\rightarrow\infty} \mathrm{P}_{\mu_k,\phi^k}(v_{\mu_k}(\tau_{0},\tau_{1})\in [\ubar{\alpha},\ubar{\alpha}+\epsilon])<1-\delta.$$
This implies that the payoff to Firm $0$ is at most $1-\delta G_0(\ubar{\alpha}+\epsilon_0)<1$.
To see this note that with a probability of at least $\delta>0$ it holds for sufficiently large $k$ that $v_{\mu_k}(\tau_{0},\tau_{1})>\ubar{\alpha}+\epsilon$. Therefore, with probability at least
$\delta$ the profit of Firm $0$ is bounded by $1-G_0(\ubar{\alpha}+\epsilon_0)$. Therefore, the expected profit of Firm $0$ is bounded by
$$\delta(1-G_0(\ubar{\alpha}+\epsilon_0))+(1-\delta)=1-\delta G_0(\ubar{\alpha}+\epsilon_0).$$

Since signals are bounded and $\lim_{k\rightarrow\infty}\mu_k=1$ it must hold, for sufficiently large $k,$ that $$2\ubar{\alpha}_{\mu_k}-1>1-\delta G_0(\ubar{\alpha}+\epsilon).$$ In the game $\Gamma(\mu_k),$ consider a deviation by  Firm $0$  to the pure price
$\tau_0=2\ubar{\alpha}_{\mu_k}-1$. Firm $0$ then guarantees an expected revenue of
$$2\ubar{\alpha}_{\mu_k}-1>1-\delta G_0(\ubar{\alpha}+\epsilon),$$ which implies a contradiction.
\end{proof}

The following corollary shows that as $\mu$ approaches $1$, it holds that for any SPE of $\Gamma(\mu)$, the equilibrium price of Firm $0$ approaches $1$.
\begin{corollary}\label{cor:price0}
Let $\{\mu_k\}_{k=0}^{\infty}\subset (0,1)$ be a sequence of priors that converges to $1,$ and let $(\phi^k_0,\phi^k_1,\sigma_k)$ be a SPE  of $\Gamma(\mu_k)$ for any $k$. Then,
$$
\lim_{k\rightarrow\infty} \phi_{0}^k=1.
$$
\end{corollary}
Corollary \ref{cor:price0} follows from Proposition \ref{prop: DE price} and Lemma \ref{lem:NDE higher price}.

The following lemma provides an upper limit to the support of Firm $1$ in a non-deterrence equilibrium.
\begin{lemma}\label{cor:price1}
If $(\phi_0,\phi_1)$ is a non-deterrence equilibrium, then $\phi_0([2\ubar{\alpha}_{\mu}-1,\uBound])=1$ and $\phi_1([1-2\uBound,1-\ubar{\alpha}_{\mu}])=1$.
\end{lemma}
\begin{proof}
It follows from Proposition \ref{prop: DE price} that $\Pi_0(\phi_0,\phi_1)>0.$
Note further that for any price $\tau_0>\uBound$ the consumer would be strictly better off choosing $e$ than buying from Firm $0$.
Therefore, we must have $\phi_0((\uBound,1])=0$ for otherwise a profitable deviation could have been constructed for Firm $0$.
\end{proof}

Finally, we present a proof of the third part of Theorem \ref{thm:SSG_results}, which considers the case of non-vanishing likelihood. In such a case, whenever the prior is sufficiently biased in favor of one firm, there is a unique equilibrium in which the a priori unfavorable firm is deterred.

\begin{proposition}\label{prop:Unique DE}
If $g_0(\ubar{\alpha})>0$, then $\exists \bar\publicBelief\in(0,1)$ such that any SPE of $\Gamma(\publicBelief)$ is a deterrence equilibrium for all $\publicBelief>\bar\publicBelief.$
 Symmetrically, if $g_1(\bar{\alpha})>0$, then $\exists \ubar\publicBelief\in(0,1)$ such that
any SPE of $\Gamma(\publicBelief)$ is a deterrence equilibrium
for all $\publicBelief<\ubar\publicBelief.$
\end{proposition}

\begin{proof}
We prove the first part of the proposition. The proof of the second part follows from symmetric considerations.

Assume by way of contradiction that there exists a sequence of priors $\{\mu_k\}$ such that $\lim_{k\rightarrow\infty}\mu_k=1$ and a corresponding sequence of SPEs,  $\{(\phi_{0}^k,\phi_{1}^k)\}_{k=1}^\infty$, such that $\phi^k=(\phi_{0}^k,\phi_{1}^k)$ is a non-deterrence equilibrium of $\Gamma(\mu_k)$ for all values of $k$.

Note that it must be the case that for almost every realized price $\tau_0$ (with respect to $\phi_0^k$) of Firm $0$
$$\Pi_0(\tau_{0},\phi_{1}^k)=\Pi_0(\phi_{0}^k,\phi_{1}^k)$$
(otherwise Firm $0$ would have a profitable deviation).

Let $\tau_0^k$ be the highest price in the support of $\phi^k_0$. It follows from the above that
\begin{equation}\label{eq:indifhh}
\Pi_0(\tau_0^k,\phi_{1}^k)=\Pi_0(\phi_{0}^k,\phi_{1}^k).
\end{equation}
Since $(\phi_0^k,\phi_1^k)$ is a non-deterrence equilibrium, Lemma \ref{lem:NDE higher price} implies that $\phi_0((2\ubar{\alpha}_{\mu_k}-1,1])>0$ for all $k\geq 1$.

We next show that for all sufficiently large $k$ there exists $\varepsilon>0$ such that $\Pi_0(\tau_0^k-\varepsilon,\phi_1^k)-\Pi_0(\tau_0^k,\phi_1^k)>0$.

We claim first that $v_{\mu_k}(\tau_0^k,\tau_1)>\ubar{\alpha}$ for almost every realized $\tau_1$ (with respect to $\phi_1^k$). Assume that there exists a measurable subset $T\subset[0,1]$
with $\phi_1^k(T)>0$ such that $v_{\mu_k}(\tau_0^k,\tau_1)=\ubar{\alpha}$ for all $\tau_0^k.$
Since $v_{\mu_k}(\tau_0,\tau_1)$ is increasing in $\tau_0$ for every fixed $\tau_1,$ it follows from the definition
of $\tau_0^k$ that $v_{\mu_k}(\tau^k_0,\tau_1)=\ubar{\alpha}$ for $\phi_0^k$ almost all realized prices $\tau^k_0$ of Firm $0$.  Therefore, we must have that the profit to Firm $1,$ conditional on $\tau_1\in T,$  is zero. By Lemma \ref{lem:NDE higher price}, Firm $1$'s expected payoff under $\phi^k$ is strictly positive, and hence we must have a profitable deviation for Firm $1$.

Using equation \eqref{eq:zero_profit} we can write
\begin{equation}\label{eq:npo}
\begin{split}
&\Diff{\Pi_{0}(\tau_0,\phi^k_1)}{\price_{0}}|_{\price_0=\tau_0^k}=\\
&\int\big(\mu_k(1-G_{0}(v_{\mu_k}(\tau_0,\tau_1)+(1-\mu_k)(1-G_{1}(v_{\mu_k}(\tau_0,\tau_1))\big)d\phi^k_1(\tau_1)\\
&-\tau_0^k\left(\Diff{v_{\mu_k}(\tau_0,\tau_1)}{\price_{0}}|_{\tau_0^k}\right)
\left(\publicBelief_k g_{0}(v_{\mu_k}(\tau_0^k,\tau_1))+(1-\publicBelief)g_{1}(v_{\mu_k}(\tau_0^k,\tau_1))\right)d\phi^k_1(\tau_1).
\end{split}
\end{equation}
Since $\lim_{k\rightarrow\infty}\mu_k=1$, it follows from Lemma \ref{cor: along the path} that $$\lim_{k\rightarrow\infty} \mathrm{P}_{\mu_k,\phi_1^k}(v_{\mu_k}(\tau_k,\tau_1^k)-\ubar{\alpha}>\delta)=0,$$
for any $\delta>0$.

Since the signal structure $(F_0,F_1,S)$ does not exhibit the  vanishing likelihood property, it follows that $g_1(\ubar{\alpha})>0$ and, by Lemma \ref{lem:g_0=0_in_VL}, that $g_0(\ubar{\alpha})>0.$ Therefore, for some $\beta>0,$
$$\lim_{k\rightarrow\infty}\mathrm{P}_{\mu_k,\phi_1^k}(\mu_k g_{0}(v_{\mu_k}(\tau_0^k,\tau_1))+(1-\publicBelief_k)g_{1}(v_{\mu_k}(\tau_0^k,\tau_1))>\beta)=1.$$
 We further note that  $\phi_1^k([0,1-\ubar{\alpha}_{\mu_k}])=1$  by Lemma \ref{cor:price1}.

 Since $\lim_{k\rightarrow\infty}\mu_k=1$ we have that $\lim_{k\rightarrow\infty}\phi^k_1=0$.
Moreover, Corollary \ref{cor:price0} implies that $\lim_{k\rightarrow\infty}\tau^k_0=\lim_{k\rightarrow\infty}\phi_0^k=1.$ Therefore, $\lim_{k\rightarrow\infty}(\tau_{0}^k-\tau_{1}^k)^2=1$.
 Hence equation \eqref{eq:treshold_function} and equation \eqref{eq:v_diff_gen1} of Lemma \ref{lem:llr} imply that
\begin{equation}\label{eq:diffginfity}
\lim_{k\rightarrow\infty}\left(\Diff{v_{\mu_k}(\tau_0,\tau^k_1)}{\price_{0}}|_{\tau_0^k}\right)=\infty,
\end{equation}
for any choice of $\tau^k_1$ in the support of $\phi^k_1$.
Therefore, equation \eqref{eq:npo} implies that $\Diff{\Pi_{0}(\tau_0,\tau_1)}{\price_{0}}|_{\price_0=\tau_0^k}<0$ for all sufficiently large values of $k$.
Hence, in particular, for all sufficiently large values of $k$ there exists a sufficiently small $\varepsilon>0$ such that
$$\Pi_{0}(\tau^k_0-\varepsilon,\phi^k_1)-\Pi_{0}(\tau^k_0,\phi^k_1)>0.$$
Therefore equation \eqref{eq:indifhh} implies that for all sufficiently large $k,$ Firm $0$ has a profitable deviation from $\phi^k_0$. This stands in contradiction
to the assumption that $(\phi^k_0,\phi^k_1)$ is an equilibrium strategy.

\end{proof}

Theorem \ref{thm:SSG_results} consolidates Corollary \ref{cor:unbounded NDE}, Proposition \ref{prop:no DE when g=0}, and Proposition \ref{prop:Unique DE}.


\section{Proofs for the Farsighted Firms}\label{sec:farsignted_proofs}

We state some lemmata that will prove useful for obtaining the results for farsighted firms.

\begin{lemma}\label{lem:td_convex}
	$2\ubar{\alpha}_\mu-1$ is a strictly convex function of $\mu$ with a bounded derivative on $[0,1]$.
\end{lemma}
\begin{proof}
	Let
	\begin{equation}
	h(\mu)=2\ubar{\alpha}_\mu-1=2\frac{\mu\ubar{\alpha}}{\mu\ubar{\alpha}+(1-\mu)(1-\ubar{\alpha})}-1.
	\end{equation}
The first derivative of $h(\mu)$ is  $h'(\mu)=\frac{2\ubar{\alpha}(1-\ubar{\alpha})}{[\mu(2\ubar{\alpha}-1)+(1-\mu)]^2},$
which is bounded by $\frac{2(1-\ubar{\alpha})}{\ubar{\alpha}}$. This establishes the first claim of the lemma.

The second derivative of $h(\mu)$ is
	\begin{equation}\label{eq:2d}
	\frac{d^2 h(\mu)}{d \mu^2} =\frac{4(1-\ubar{\alpha})\ubar{\alpha}(2\ubar{\alpha}-1)}{(\mu(1-2\ubar{\alpha})-(1-\ubar{\alpha}))^3}.
	\end{equation}
Recall that $\ubar{\alpha}<0.5$ and so the enumerator in equation \eqref{eq:2d} is negative. In addition,
as $\mu\leq1$ we conclude that $1-\ubar{\alpha}>1-2\ubar{\alpha}\geq\mu(1-2\ubar{\alpha})$ and so  the denominator of \eqref{eq:2d} is also negative. Thus,  $\frac{d^2 h(\mu)}{d \mu^2} >0$ and so $h(\mu)$ must be strictly convex.
\end{proof}

\subsection{Proof of Theorem \ref{prop:non-myopic-VL-NDE}}
Assume by way contradiction that signals exhibit vanishing likelihood and asymptotic learning fails in some equilibrium $(\bar{\phi}_0,\bar{\phi}_1,\bar{\sigma})$.
This entails that for some constant $0<\eta<\frac{1}{2}$ one of the two following events, $\mu_\infty\in [\frac{1}{2},1-\eta]$ or $\mu_\infty\in [\eta,\frac{1}{2}]$, has positive probability.  Let us assume, without loss of generality, that it is the former event. Let $r>0$ denote its probability. More formally, if $H'\subseteq H_\infty$ is the set of all infinite histories for which
$\lim_{t\rightarrow\infty} \mu_t=\mu_\infty\in [\frac{1}{2},1-\eta]$ then $\mathbf{P}_{(\bar{\phi}_0,\bar{\phi}_1,\bar{\sigma})}(H')= r >0$.
Further note that for $\mathbf{P}_{(\bar{\phi}_0,\bar{\phi}_1,\bar{\sigma})}$-almost every infinite history $h\in H'$ it must hold that
$\lim_{t\rightarrow\infty}\Pi^\delta_1(\bar{\phi}_0,\bar{\phi}_1,\bar{\sigma}|h_t)=0$ ($h_t\in H_t$ is the truncation of $h$).

Using similar arguments to those invoked in the proof of Proposition \ref{prop: DE price} we conclude that for almost every history
$h\in H'$ it holds that
$\lim_{t\rightarrow\infty}|\Pi^\delta_0(\bar{\phi}_0,\bar{\phi}_1,\bar{\sigma}|h_t)-2\ubar{\alpha}_{\mu_t}-1|=0$.

For every $\epsilon>0$ let us denote by $H_\epsilon\subseteq H$ the set of all finite histories for which
$$|\Pi^\delta_0(\bar{\phi}_0,\bar{\phi}_1,\bar{\sigma}|h_t)-2\ubar{\alpha}_{\mu_t}-1|\leq\epsilon.$$
 It follows from the above that for every $\epsilon>0,$
 $$\mathbf{P}_{(\bar{\phi}_0,\bar{\phi}_1,\bar{\sigma})}(\exists t \text{ s.t. } h_t\in H_\epsilon\text{ and }\mu_t\leq 1-\frac{\eta}{2})\geq r.$$

As we assume vanishing likelihood, we can invoke Proposition \ref{prop:no DE when g=0} and conclude that there exists sufficiently small $\epsilon_0>0$ such that whenever $\mu\leq 1-\frac{\eta}{2},$ Firm $0$ has some price $\tau_0\in [0,1]$ that guarantees the following stage payoff:
\begin{equation}\label{eq:pd0}
\Pi_{0}(\tau_0,\phi_1,\publicBelief)>2\ubar{\alpha}_{\mu}-1+\epsilon_0,
\end{equation}
for every strategy $\phi_1\in\Delta([0,1])$ of Firm $1$.

We define the strategy $\hat{\phi}_0$ for Firm $0$ as follows. Let $h$ be some finite history.
If $h \in H_{\frac{\epsilon_0}{2}}$ then set $\hat{\phi}_0(h_t) = \tau_0$, where $\tau_0$ is the price that satisfies the inequality in equation \eqref{eq:pd0}.
If  $h \not \in H_{\frac{\epsilon_0}{2}}$ but has some prefix $h_t\in H_{\frac{\epsilon_0}{2}}$ then set the price at
 $2\ubar{\alpha}_{\mu_{t+1}}-1$ (where $\mu_{t+1}$ is the public belief at stage $t+1$). Note that this implies that from stage $t+1$ onward Firm $1$ is deterred and the public belief remains fixed.
 Finally, whenever no prefix of $h$ is in $H_{\frac{\epsilon_0}{2}}$ let the price be that which was chosen according to the original strategy $\bar{\phi}_0$.

The continuation payoff of Firm $0$ for any finite history in $H_{\frac{\epsilon_0}{2}}$ is $$(1-\delta)[2\ubar{\alpha}_{\mu_t}-1+\epsilon_0]+\delta E_{(\hat{\phi}_0,\bar{\phi}_1,\bar{\sigma})}[2\ubar{\alpha}_{\mu_{t+1}}-1|h_t].$$
That is, the current period deviation of Firm $0$ yields, by equation \eqref{eq:pd0}, an expected payoff of at least $2\ubar{\alpha}_{\mu_t}-1+\epsilon_0$. Thereafter the value of $\mu_{t+1}$ is realized  and in all subsequent periods $t'>t$ Firm $0$ receives a constant payoff of $2\ubar{\alpha}_{\mu_{t+1}}-1$. As the function $2\ubar{\alpha}_{\mu}-1$ is convex (by Lemma \ref{lem:td_convex}) this continuation payoff is guaranteed to satisfy the following inequality:
$$E_{(\hat{\phi}_0,\bar{\phi}_1),\bar{\sigma}}[2\ubar{\alpha}_{\mu_{t+1}}-1|h_t]\geq 2\ubar{\alpha}_{\mu_{t}}-1.$$
Comparing this with the continuation payoff from the original strategy implies that the deviation yields a profit that is at least $(1-\delta)\frac{\epsilon_0}{2}$. This stands in contradiction to the fact that $(\bar{\phi}_0,\bar{\phi}_1,\bar{\sigma})$ is a Bayesian Nash equilibrium.

\subsection{Proof of Theorem \ref{prop:non-myopic-lmpe-VL-NDE}.}

Assume that the signal structure does not satisfy the vanishing likelihood property and so either $g_0(\ubar{\alpha})>0$ or $g_0(\bar{\alpha})>0$. Let us assume without loss of generality that $g_0(\ubar{\alpha})>0$.
We first prove the following lemma.
\begin{lemma}\label{lem:threshold}
For every discount factor $\delta\in(0,1)$ there exists $\mu'>0$ such that if $\mu_0\geq\mu'$, then the following strategy is a Bayesian Nash equilibrium strategy of the repeated game with initial prior $\mu_0$ and  discount factor $\delta$:
at every history $h_t$ Firm $0$ asks a price of $2\ubar{\alpha}_{\mu_t}-1$ and Firm $1$ asks a price $0$.
\end{lemma}
\begin{proof}
Firm $0$ sets a price of $2\ubar{\alpha}_{\mu_t}-1$ that necessarily drives Firm $1$ out of the market. Therefore, no matter what price Firm $1$ offers it cannot make a positive profit and so no profitable deviation exists for it. To establish the claim we need to show that if $\mu_0$ is sufficiently large, then Firm $0$ also doesn't have a profitable deviation.

To see this, we show next that for a large enough initial prior $\mu$ (we omit the subscript $0$ that refers to the initial stage), Firm $0$ does not have a one-shot profitable deviation given this strategy profile. In fact, it is enough to establish this for the first stage.

Whenever Firm $0$ increases its price to $\tau,$ and the price profile becomes $\tau=(\tau,0)$  (clearly decreasing the price is not profitable), there are two possibilities that can happen in the first stage: either the consumer buys from Firm $0$, in which case we denote new public belief $\mu_{0}(\mu,\tau)$, or the consumer buys from Firm $1,$ in which case we denote new public belief $\mu_{1}(\mu,\tau)$. With this new notation we formulate the continuation payoff from a one-shot deviation at the first stage, as a function of the new price $\tau$:

\begin{align*}
&(1-\delta)(1-G_\mu(v_\mu(\tau,0))\tau+\\ \notag &\delta[(1-G_\mu(v_\mu(\tau,0)))(2\ubar{\alpha}_{\mu_{0}(\mu,\tau)}-1)+G_\mu(v_\mu(\tau,0))(2\ubar{\alpha}_{\mu_{1}(\mu,\tau)}-1)].
\end{align*}

We need to show that for all sufficiently large $\mu$ this continuation payoff is at most the payoff prior to a deviation, $2\ubar{\alpha}_{\mu}-1,$  for every $\tau$. Equivalently, we need to show that
\begin{align}\label{eq:contin}
&(1-\delta)[2\ubar{\alpha}_{\mu}-1-(1-G_\mu(v_\mu(\tau,0)))\tau]\geq \\
\notag &\delta[(1-G_\mu(v_\mu(\tau,0)))(2\ubar{\alpha}_{\mu_{0}(\mu,\tau)}-2\ubar{\alpha}_{\mu})+G_\mu(v_\mu(\tau,0))(2\ubar{\alpha}_{\mu_{1}(\mu,\tau)}-2\ubar{\alpha}_{\mu})]
\end{align}

Using Bayes, rule we infer the following three equalities:
$$
1-G_\mu(v_\mu(\tau,0))=\mu\bigg(1-G_{0}\big(v_\mu(\tau,0)\big)\bigg)+
(1-\mu)\bigg(1-G_{1}\big(v_\mu(\tau,0)\big)\bigg),
$$
$$
\mu_{0}(\mu,\tau)=\frac{\mu (1-G_0(v_\mu(\tau,0)))}{1-G_\mu(v_\mu(\tau,0))},
$$
and
$$
\mu_{1}(\mu,\tau)=\frac{\mu G_0(v_\mu(\tau,0))}{G_\mu(v_\mu(\tau,0))}.
$$

A simple calculation yields
\begin{align}
&\mu_{0}(\mu,\tau)-\mu=\frac{\mu(1-\mu)(G_1(v_\mu(\tau,0))-G_0(v_\mu(\tau,0)))}{1-G_\mu(v_\mu(\tau,0))},\label{eq:sstep23}\\
&\mu-\mu_{1}(\mu,\tau)=\frac{\mu(1-\mu)(G_1(v_\mu(\tau,0))-G_0(v_\mu(\tau,0)))}{G_\mu(v_\mu(\tau,0))}.\label{eq:sstep24}
\end{align}

In addition, since the function $\ubar{\alpha}_{\mu}$ is convex in $\mu$ with a bounded derivative (by Lemma \ref{lem:td_convex}), it is also a Lipschitz function of $\mu$. Therefore, there exists a constant $C$ such that
\begin{align}
&0<2\ubar{\alpha}_{\mu_{0}(\mu,\tau)}-2\ubar{\alpha}_{\mu}\leq C(\mu_{0}(\mu,\tau)-\mu),\label{eq:sstep23a}\\
&0<2\ubar{\alpha}_{\mu}-2\ubar{\alpha}_{\mu_{1}(\mu,\tau)}\leq C(\mu-\mu_{1}(\mu,\tau)).\label{eq:sstep24a}
\end{align}


Instead of proving inequality \eqref{eq:contin},
by equations \eqref{eq:sstep23}, \eqref{eq:sstep24}, \eqref{eq:sstep23a}, and \eqref{eq:sstep24a} it is sufficient to prove that for all sufficiently large $\mu$ it holds for every $\tau\in[0,1]$ that
\begin{align}\label{eq:step31}
&(1-\delta)[2\ubar{\alpha}_{\mu}-1+(1-G_\mu(v_\mu(\tau,0)))\tau]\geq\\
&\notag [G_1(v_\mu(\tau,0))-G_0(v_\mu(\tau,0))] \delta\mu(1-\mu)C.
\end{align}
Assume by way of contradiction that this is not the case. Then there exists a sequence of priors $\{\mu_k\}_{k=1}^\infty\subset (0,1)$ with a limit $1$ and an appropriate sequence
$\{\tau_k\}_{k=1}^\infty$ of prices such that
\begin{align}\label{eq:step312}
&(1-\delta)[2\ubar{\alpha}_{\mu_k}-1-(1-G_{\mu_k}(v_{\mu_k}(\tau_k,0)))\tau_k]<\\
\notag&[G_1(v_{\mu_k}(\tau_k,0))-G_0(v_{\mu_k}(\tau_k,0))] \delta\mu_k(1-\mu_k)C.
\end{align}

Since $v_{\mu_k}(\tau_k,0)>\ubar{\alpha}$ (for otherwise we have an equality in \eqref{eq:step312}) we can rewrite \eqref{eq:step312} as follows:
\begin{align}\label{eq:aeq1}
\frac{2\ubar{\alpha}_{\mu_k}-1-(1-G_{\mu_k}(v_{\mu_k}(\tau_k,0)))\tau_k}{G_1(v_{\mu_k}(\tau_k,0))-G_0(v_{\mu_k}(\tau_k,0))}<
\frac{\delta}{1-\delta}\mu_k(1-\mu_k)C.
\end{align}
We rearrange \eqref{eq:aeq1} and get
	\begin{equation}\label{eq:step32}
	\begin{split}
	 &\frac{2\ubar{\alpha}_{\mu_k}-1-\tau_k}{G_1(v_{\mu_k}(\tau_k,0))-G_0(v_{\mu_k}(\tau_k,0))}+\frac{G_1(v_{\mu_k}(\tau_k,0))\tau_k}{G_1(v_{\mu_k}(\tau_k))-G_0(v_{\mu_k}(\tau_k))} < \frac{\delta}{1-\delta}\mu_k(1-\mu_k)C.
	\end{split}
	\end{equation}
Clearly, since $\lim_{k\rightarrow\infty}\mu_k=1,$ the limit of the right-hand side of \eqref{eq:step32} is zero and therefore in order to establish contradiction we show that
\begin{equation}\label{eq:step33}
	\begin{split}
	 &\frac{2\ubar{\alpha}_{\mu_k}-1-\tau_k}{G_1(v_{\mu_k}(\tau_k,0))-G_0(v_{\mu_k}(\tau_k,0))}+\frac{G_1(v_{\mu_k}(\tau_k,0))\tau_k}{G_1(v_{\mu_k}(\tau_k))-G_0(v_{\mu_k}(\tau_k))}
	\end{split}
	\end{equation}
is bounded away from zero as $k$ goes to infinity.	

We claim first that $\liminf_{k\rightarrow\infty}\frac{G_1(v_{\mu_k}(\tau_k))\tau_k}{G_1(v_{\mu_k}(\tau_k))-G_0(v_{\mu_k}(\tau_k))}\geq 1$.
To do so we recall equation \eqref{eq:step312}.
Note that as  $k$ goes to infinity, $1-\mu_k$ goes to $0$ and so the right-hand side of \eqref{eq:step312} goes to $0$. In addition, $\lim_{k\rightarrow\infty}2\ubar{\alpha}_{\mu_k}-1=1$ and so equation \eqref{eq:step312} implies that
$
\lim_{k\gor\infty}(1-G_{\mu_k}(v_{\mu_k}(\tau_k,0)))\tau_k \ge 1.
$
As both multipliers are in the unit interval this necessarily implies that both approach $1$.
Therefore, $\lim_{k\rightarrow\infty}\tau_k=1$ and
$\lim_{k\gor\infty}G_{\mu_k}(v_{\mu_k}(\tau_k,0))=0$. The latter then implies that
\begin{equation}\label{eq:gtll}
\lim_{k\gor\infty}v_{\mu_k}(\tau_k,0)=\ubar{\alpha}.
\end{equation}

We note that since $v_{\mu_k}(\tau_k,0)>\ubar{\alpha},$ Lemma \ref{lem:aux} implies that $1>G_1(v_{\mu_k}(\tau_k))-G_0(v_{\mu_k}(\tau_k))>0$.
Therefore ${G_1(v_{\mu_k}(\tau_k))\tau_k}{G_1(v_{\mu_k}(\tau_k))-G_0(v_{\mu_k}(\tau_k))}\geq 1$. Furthermore, since $\lim_{k\rightarrow\infty}\tau_k=1$ it holds that
\begin{equation}\label{eq:bafi1}
\liminf_{k\rightarrow\infty}\frac{G_1(v_{\mu_k}(\tau_k))\tau_k}{G_1(v_{\mu_k}(\tau_k))-G_0(v_{\mu_k}(\tau_k))}\geq 1.
\end{equation}

Therefore in order to show that the term in equation \eqref{eq:step33} is bounded away from zero and hence leads to a contradiction, it is sufficient to show that
\begin{equation}\label{eq:bafi}
\frac{2\ubar{\alpha}_{\mu_k}-1-\tau_k}{G_1(v_{\mu_k}(\tau_k,0))-G_0(v_{\mu_k}(\tau_k,0))}
\end{equation}
has a limit $0$ as $k$ goes to infinity, which establishes a contradiction.

By a standard first-order approximation
	it follows that for every $\epsilon>0$ there exists $x_\epsilon>\ubar{\alpha}$ such that for every $x\leq x_\epsilon$,
	$$(x-\ubar{\alpha})[g_i(\ubar{\alpha})-\epsilon] \leq G_i(x)\leq (x-\ubar{\alpha})[g_i(\ubar{\alpha})+\epsilon].$$
Since $\tau_k>2\ubar{\alpha}_{\mu_k}-1$, the numerator in \eqref{eq:bafi} is negative. By Lemma \ref{lem:aux} the denominator is positive. We use the approximation to get
	\begin{equation}
	\frac{2\ubar{\alpha}_{\mu_k}-1-\tau_k}{G_1(v(\tau_k,0))-G_0(v(\tau_k,0))}>
	 \frac{2\ubar{\alpha}_{\mu_k}-1-\tau_k}{(g_1(\ubar{\alpha})-g_0(\ubar{\alpha})+2\varepsilon)(v_{\mu_k}(\tau_k,0)-\ubar{\alpha})}.
	\end{equation}
	Now recall that $v_{\mu_k}(\tau_k,0)-\ubar{\alpha}=v_{\mu_k}(\tau_k,0)-v_{\mu_k}(2\ubar{\alpha}_{\mu_k}-1,0)=v'_{\mu_k}(\tilde{\tau}_k,0)(\tau_k-2\ubar{\alpha}_{\mu}+1)$ for some $\tilde{\tau}_k\in(2\ubar{\alpha}_{\mu}-1,\tau_k)$. Therefore,
\begin{equation}
\frac{2\ubar{\alpha}_{\mu_k}-1-\tau_k}{G_1(v_{\mu_k}(\tau_k,0))-G_0(v_{\mu_k}(\tau_k,0))}>\frac{-1}{(g_1(\ubar{\alpha})-g_0(\ubar{\alpha})+2\varepsilon)v'_{\mu_k}(\tilde{\tau}_k,0)}.
\end{equation}
Since $\tilde{\tau}_k\in(2\ubar{\alpha}_{\mu}-1,\tau_k)$ we have that $\lim_{k\rightarrow\infty} \tilde{\tau}_k=1$. Since Firm $1$ offers its product at price zero, the market is full.  Therefore, Lemma \ref{lem:llr} implies that $\lim_{k\rightarrow\infty}v'_{\mu}(\tilde{\tau}_k,0)=\infty$. As a result,  \eqref{eq:bafi} approaches zero as $k$ goes to infinity.
This establishes a contradiction and concludes the proof of Lemma \ref{lem:threshold}.

\end{proof}
We now conclude the proof of Theorem \ref{prop:non-myopic-lmpe-VL-NDE}.
\begin{proof}[\textbf{Proof of Theorem \ref{prop:non-myopic-lmpe-VL-NDE}}]
Fix a discount factor $\delta\in (0,1)$. By Lemma \ref{lem:threshold} there exists a $\mu'<1$ such that the strategy where at every history $h_t$ Firm $0$ asks a price of $2\ubar{\alpha}_{\mu_t}-1$ and Firm $1$ asks a price of $0$ is a Bayesian Nash equilibrium of the game with an initial prior $\mu_0\geq \mu'$. Consider an arbitrary prior $\mu_0$. We define a game where whenever $\mu_t\geq \mu'$ Firm $0$ and Firm $1$ are constrained to play actions $2\ubar{\alpha}_{\mu_t}-1$ and $0$ respectively with probability one. It is easy to see that this game possesses a Bayesian equilibrium $(\bar{\tau}_0,\bar{\tau}_1,\bar{\sigma})$. By Lemma \ref{lem:threshold} this equilibrium is also a Bayesian Nash equilibrium of the unconstrained game. Clearly, under $(\bar{\phi}_0,\bar{\phi}_1,\bar{\sigma})$ asymptotic learning must fail. This concludes the proof of Theorem  \ref{prop:non-myopic-lmpe-VL-NDE}.

\end{proof}
	
\section{Auxilliary Lemmas}\label{sec:Acem}
\begin{lemma}\label{lem:aux}
The ratio $\frac{G_1(r)}{G_0(r)}$ is non-increasing in $r$ and $\frac{G_1(r)}{G_0(r)}>1$ for all $r\in(\ubar{\alpha},\bar{\alpha}).$
In particular, $G_0$ first-order stochastically dominates $G_1$.
\end{lemma}
\begin{proof}
The proof follows from the more general result that appears in Lemma A1 of \cite{Acemoglu2011}.
\end{proof}
\begin{cor2}
Let $(\overline{\sigma},\overline{\price}_0,\overline{\price}_1)$ be a myopic Bayesian Nash equilibrium. If asymptotic learning holds, then conditional on state $\omega\in\Omega$,
$$\lim_{t\rightarrow\infty} \mathbf{P}_{(\overline{\sigma},\overline{\price}_0,\overline{\price}_1)}(\{\sigma_{t}(\mu_t,s,\overline\tau(\mu_t))=\omega\}|\omega)=1.$$
\end{cor2}
\begin{proof}
Without loss of generality assume that the realized state is $\omega=0$.  Since asymptotic learning holds, we have that $\lim_t\mu_t=1$ almost surely. By Corollary \ref{cor: along the path} we have that  $\lim_{t\gor\infty}v_{\mu_t}(\tau_{0}^t(\mu_t),\tau_{1}^t(\mu_t))=\ubar{\alpha}$. Therefore,
\begin{align*}
& \lim_{t\rightarrow\infty} \mathbf{P}_{(\overline{\sigma},\overline{\price}_0,\overline{\price}_1)}(\{\sigma_{t}(\mu_t,s,\overline\tau(\mu_t))=0\}|\omega=0)= \lim_{t\rightarrow\infty}G_0(v_{\mu_t}(\tau_{0}^t(\mu_t),\tau_{1}^t(\mu_t)))=\\& G_0(\ubar{\alpha})=1.
\end{align*}

\end{proof}

\end{document}